\documentclass[10pt,journal,compsoc]{IEEEtran}

\usepackage{amsmath,amssymb,amsfonts,amsthm}
\usepackage[linesnumbered,ruled,vlined]{algorithm2e}
\SetArgSty{textup}
\usepackage{array}
\usepackage{bbding}
\usepackage{bm}
\usepackage{booktabs}
\usepackage{enumitem}
\usepackage{graphicx}
\usepackage{hyperref}
\usepackage{listing}
\usepackage{makecell}
\usepackage{multirow}
\usepackage{pifont}
\usepackage{subcaption}
\usepackage{threeparttable}
\usepackage[most]{tcolorbox}

\newtcolorbox{mybox}[2][]{%
	attach boxed title to top left
	= {xshift=10pt, yshift=-10pt},
	colframe     = black,
	colback      = white,
	coltitle     = black,
	colbacktitle = white,
	title        = #2,#1,
	enhanced,
	left=2pt,
	top=10pt,
}
\usepackage{textcomp}
\newtheorem{theorem}{Theorem}

\usepackage{tikz}
\usetikzlibrary{positioning}
\usepackage[normalem]{ulem}
\usepackage{xcolor}
\newtheorem{definition}{Definition}[section]

\newcommand{\xmark}{\ding{55}}%

\ifCLASSOPTIONcompsoc
\usepackage{cite}
\else
\usepackage{cite}
\fi

\ifCLASSINFOpdf

\else

\fi

\begin{document}
	
	\title{FileDAG: A Multi-Version Decentralized Storage Network Built on DAG-based Blockchain}

	\author{Hechuan Guo,
		Minghui Xu,
		Jiahao Zhang,
		Chunchi Liu,
		Dongxiao Yu,
            Schahram Dustdar,
            Xiuzhen Cheng
    
		\IEEEcompsocitemizethanks{\IEEEcompsocthanksitem H. Guo, M. Xu, J. Zhang, D. Yu and X. Cheng are with the School of Computer and Science and Technology, Shandong University. Email: \{ghc, zjh\}@mail.sdu.edu.cn, \{mhxu, dxyu, xzcheng\}@sdu.edu.cn\\ 
			
		\IEEEcompsocthanksitem C. Liu was with the Department of Computer Science, The George Washington University and now with Ernst \& Young. E-mail: liuchunchi@gwu.edu\\
   
            \IEEEcompsocthanksitem S. Dustdar is with the Research Division of Distributed Systems, TU Wien. Email: dustdar@dsg.tuwien.ac.at\\
            
            Corresponding author: Minghui Xu.
            }
	}

	\IEEEtitleabstractindextext{%
		\begin{abstract}
			Decentralized Storage Networks (DSNs) can gather storage resources from mutually untrusted providers and form worldwide decentralized file systems. Compared to traditional storage networks, DSNs are built on top of blockchains, which can incentivize service providers and ensure strong security. However, existing DSNs face two major challenges. First, deduplication can only be achieved at the directory-level. Missing file-level deduplication leads to unavoidable extra storage and bandwidth cost. Second, current DSNs realize file indexing by storing extra metadata while blockchain ledgers are not fully exploited. To overcome these problems, we propose FileDAG, a DSN built on DAG-based blockchain to support file-level deduplication in storing multi-versioned files. When updating files, we adopt an increment generation method to calculate and store only the increments instead of the entire updated files. Besides, we introduce a two-layer DAG-based blockchain ledger, by which FileDAG can provide flexible and storage-saving file indexing by directly using the blockchain database without incurring extra storage overhead. We implement FileDAG and evaluate its performance with extensive experiments. The results demonstrate that FileDAG outperforms the state-of-the-art industrial DSNs considering storage cost and latency.
		\end{abstract}

		\begin{IEEEkeywords}
			Decentralized storage networks, DAG-based blockchain, deduplication, file indexing
	\end{IEEEkeywords}}
	
	\maketitle
	
	\IEEEdisplaynontitleabstractindextext
	
	\IEEEpeerreviewmaketitle

	\section{Introduction}\label{sec:introduction}
	
	Blockchain technology is allowing for decentralized computing by creating decentralized trust. This has led to the development of various trustless applications, such as decentralized learning \cite{spdl}, trusted IoT data collection \cite{tems}, and blockchain-based cloud services \cite{cloudchain}. To improve the performance and efficiency of decentralized computing, decentralized storage networks have been created to decrease the amount of redundant storage needed for blockchain. Decentralized Storage Networking is an emerging technology that can aggregate free storage spaces offered by independent storage providers and self-coordinate to provide data storage and retrieval services. Compared to traditional storage networks \cite{bittorrent, storageareanetwork}, a decentralized storage network (DSN) is operated on a blockchain system, which works as an incentive layer. Blockchain rewards miners who provide reliable storage to clients, and thus enables an open manageable storage market. Besides, blockchain can act as a state machine replication protocol to ensure the consistency of file storage against Byzantine nodes. Leveraging blockchain technologies, DSNs (e.g., Filecoin \cite{filecoin}, Storj \cite{storj}, Sia \cite{sia}, Swarm \cite{swarm}) provide worldwide, robust and secure storage services among mutually untrusted users. Filecoin, as the most popular DSN, was built on top of InterPlanetary File System (IPFS) and adopts a novel proof-of-replication method proving that data is correctly stored. As storage infrastructures, DSNs have demonstrated their advantages in applications such as Web 3.0 \cite{w3s}, data sharing \cite{btfs}, and content delivery \cite{filebase}. However, current DSN schemes overlook the following two problems, which significantly affect their performance. 

	\noindent\textbf{[P1] Deduplication in multi-versioned files.}
	Supporting multi-versioned file storage is necessary in DSNs since files are usually dynamically changed or edited and users need to query different versions of a file from time to time. However, files on current DSNs are not editable. Users have to upload all versions of a file, resulting in high redundancy. Even though some DSNs have made efforts in supporting directory-level deduplication, they cannot avoid fine-grained file-level redundancy. For example, Filecoin realizes directory-level deduplication using Merkle DAG \cite{merkledag}, in which objects including files, file chunks, and directories are organized into a Merkle DAG based on their nested relationships, to remove duplicated objects among different directories; but redundencies among different file versions are still unavoidable. Missing file-level deduplication causes the waste of storage and bandwidth. Nevertheless, achieving file-level deduplication is challenging. Due to encryption and obfuscation applied on files, correlation among different versions is implicit in current DSNs, making it very hard to find duplicated contents and establish relationships among multiple versions.

    \noindent\textbf{[P2] File indexing with blockchain.}
	Traditional version control systems commonly adopt a DAG-based version graph \cite{versiongraph} to describe the relationships of multiple versions (also called derivative relationships) and help establish file indexing. However, such a graph should be maintained by a centralized server, e.g., Github \cite{githubgraph}. In current DSNs, a centralized server is not available, and the blockchain database simply stores file information in serialized transactions regardless of the derivative relationships. Therefore, it is unavoidable for each user to locally maintain an isolated database for additional metadata (e.g., a version graph) to ease its file manipulations such as create, version query, modification, merge, and fork. Furthermore, such a deficiency prevents a blockchain from serving many data-intensive applications since it takes a large amount of storage but cannot directly answer file queries in many cases, making itself like a ``burden''. 
	
	To address these problems, we propose FileDAG, a DSN system built on top of a DAG-based blockchain. FileDAG makes use of an increment-based storage mechanism to realize file-level deduplication in storing multi-versioned files. We apply increment generation algorithms to calculate the increment, an edit script that can transform a file from its previous version to its current version. Storing increments achieves fine-grained deduplication at file-level and saves storage space.
	Besides, we adopt a two-layer DAG-based blockchain ledger, which organizes transactions according to their derivative relationships. Facilitated with this ledger, one can manipulate files without establishing additional databases. Moreover, DAG-based ledgers can provide higher concurrency compared to chain-based ones \cite{tangle, hashgraph}, making FileDAG being able to handle simultaneous queries. With these design considerations, FileDAG achieves low storage cost, high system throughput, and efficient file indexing. 
	
	To validate the performance of FileDAG, we build a full-fledged FileDAG over Filecoin, by implementing the increment mechanism and the two-layer ledger mentioned above. Such an implementation ensures that FileDAG not only inherits all the nice features of Filecoin but also extends Filecoin's functionality to include effective file-level deduplication and efficient file indexing. FileDAG is a practical system possessing industrial-grade performance, as is Filecoin. To broaden the application of and welcome examinations on FileDAG, we open-source our designs at GitHub (the two layer ledger: https://github.com/zhuaiballl/DAG-Rider; increment module:
https://github.com/zhuaiballl/dyaic).
	
	\noindent \textbf{Contributions.} Compared to the existing works, our unique contributions can be summarized as follows:
	\begin{enumerate}
		\item To our best knowledge, FileDAG is the first DSN that supports file-level deduplication for multi-versioned files. We introduce an increment generation method to calculate and store the increment between two neighboring versions rather than storing the entire new version. This significantly reduces the storage cost and bandwidth usage caused by dynamical file changes. 
		\item To support file indexing, FileDAG adopts a two-layer DAG-based blockchain ledger. The lower layer supports operations including create, update, merge and fork while the upper layer ensures ledger consistency. This design integrates version graphs with a DAG-based ledger, thereby saving extra storage space for file indexing. 
		\item Finally, we provide a practical full-fledged implementation of FileDAG and evaluate its performance with extensive experiments. The results demonstrate that FileDAG outperforms the state-of-the-art DSNs considering storage cost as well as the latency of put and get operations. 
	\end{enumerate}
	
	\noindent \textbf{Organization of the paper.}
	The rest of this paper is organized as follows. Section~\ref{sec:related} summarizes related works and presents preliminary knowledge. Section~\ref{sec:filedag_design} details our FileDAG design and demonstrates how it works. Key properties and performance evaluation results of FileDAGE are respectively reported in Section~\ref{sec:analysis} and Section~\ref{sec:exp}. Finally, we summarize this paper in Section~\ref{sec:conclusion} and discuss our future research.

	\section{Related Work and Preliminaries}
	\label{sec:related}
	\subsection{Related Work}
	\subsubsection{Decentralized Storage Network}
	Filecoin \cite{filecoin}, developed by Protocol Labs, is a DSN built on top of IPFS \cite{ipfs}. It proposes Expected Consensus to adjust the winning probability of a miner based on the quantity and quality of its provided storage. Filecoin generates a hash-based content identifier (CID) for each file object (a file, a file chunk, or a directory), and allows users to reuse existing file objects for avoiding duplicatively storing them. Besides, CIDs form a Merkle DAG depicting the nested relationship of the file objects. To realize block concurrency, Filecoin introduces tipset, which allows multiple blocks to be confirmed at the same block height.
	Storj \cite{storj} and Swarm \cite{swarm} were developed based on Ethereum \cite{ethereum}. They make use of Proof-of-Stake consensus\footnote{Since September 15th, 2022, Ethereum has switched its consensus protocol from Proof-of-Work to Proof-of-Stake} and a chain-based ledger that doesn't support concurrency of blocks. Storj employs Object keys as globally unique identifiers of its file objects while Swarm generates addresses as identifiers for file chunks. These two DSN systems both achieve directory-level deduplication.
	Sia \cite{sia} adopts PoW as its consensus protocol. It builds a Merkle tree for each file and takes the Merkle root hash as the identifier of the file, thus supporting directory-level deduplication. The ledger structure of Sia is a chain, thus it cannot process blocks concurrently. Besides, Sia employs the Threefish \cite{threefish} algorithm to encrypt files, making it difficult to support version indexing.
	
	\subsubsection{File Indexing}
	
	File indexing is the process of mapping files with identifiers that can be efficiently searched. Centralized storage systems employ extra databases to record identifiers that are mapped to the locations of the corresponding files \cite{gfs}. Traditional distributed storage networks typically use distributed hash tables to realize file indexing \cite{bittorrent, gnutella, coralcdn}. Decentralized storage networks, i.e., DSNs, use content addressing technology based on distributed hash tables for file indexing. Nevertheless, current methods in DSNs are not sufficiently effective as the complete derivative relationships are hardly retained. For example, IPLD \cite{ipfs} is a data model adopted in Filecoin to describe a file or a directory as an aggregate of components linked together. With IPLD, each file is mapped to a unique hash-based identifier, and the identifier of a directory is a hash of the directory contents combined with pointers to the files. By this way, IPLD links files and directories together for file indexing. Adding, removing, or changing a file under a directory result in a different identifier of the directory, and the derivative relationship between two versions of the directory can be inferred because the two identifiers carry the pointers to the files that stay unchanged. But unfortunately IPLD fails to depict the derivative relationships among different versions of a  file. 
	Additionally, to achieve verifiability and immutability, files in DSNs are always encrypted and then made public; thus their metadata is unaccessible without a valid secret key, making file indexing a challenging problem. Based on the above analysis, one can see that the current file indexing approaches are immature and inefficient.
    
    File provenance requires to track the derivation history of a file based on file indexing. Muniswamy-Reddy \textit{et al.} \cite{pass} designed a storage system that can automatically collect and maintain provenance data. They claimed that provenance data should be maintained separately to serve different purposes. Provchain \cite{provchain} embeds the provenance data into blockchain transactions to improve efficiency and avoid additional storage cost. This incentives us to make blockchain undertake more responsibility in file indexing. 
    
    \subsubsection{Summary}
	
	\begin{table}[htbp]
	\begin{threeparttable}
		\caption{Comparison of FileDAG with Existing DSNs}
			\begin{tabular}{l c c c c}
				\toprule[1pt]
				 & \multicolumn{1}{c}{\textbf{\begin{tabular}[c]{@{}c@{}}Consensus\\ Algorithm\end{tabular}}} & \textbf{Ledger} & \multicolumn{1}{c}{\textbf{\begin{tabular}[c]{@{}c@{}}On-Chain\\ DR\end{tabular}}} & \multicolumn{1}{c}{\textbf{\begin{tabular}[c]{@{}c@{}}Deduplication\\ Level\end{tabular}}} \\
				\midrule[0.5pt]
				
				Filecoin\cite{filecoin} & \multicolumn{1}{c}{\begin{tabular}[c]{@{}c@{}}Expected\\ consensus\end{tabular}} & \multicolumn{1}{c}{\begin{tabular}[c]{@{}c@{}}DAG\\ (tipset)\end{tabular}} & \xmark & Directory \\
				
				Storj\cite{storj} & PoW & Chain & \xmark & Directory \\
				
				Sia\cite{sia} & PoW & Chain & \xmark & Directory \\
				
				Swarm\cite{swarm} & PoW & Chain & \xmark & Directory  \\
				
				\textbf{FileDAG} & DAG-Rider$^\dag$ & DAG & \checkmark & File \\
				\bottomrule[1pt]
			\end{tabular}
			\label{tab1}
		\begin{tablenotes}
		    \item[DR] Derivative Relationship
		    \item[$\dag$] Modified
		\end{tablenotes}
	\end{threeparttable}
	\end{table}
	
	A summary on the major adopted technologies and properties of FileDAG and existing DSNs is reported in Table~\ref{tab1}. One can see that current DSNs (e.g. \cite{filecoin, storj, sia, swarm}) only achieve directory-level deduplication, which means that only files can be reused but the common contents shared by different versions of a file are still stored redundantly. Missing file-level deduplication leads to the waste of storage and bandwidth. Additionally, Storj and Swarm built on Ethereum adopt a chain-based ledger, which stores transactions regardless of their derivative relationships. Sia doesn't consider storing derivative relationship between files either. Filecoin packs multiple blocks in a tipset and still ignores on-chain derivative relationships. Lacking a depiction on the complete derivative relationships among files render these systems fail to provide effective file indexing.
	
	\subsection{Preliminaries}
	\label{sec:pre}
	In this subsection, we provide the preliminary knowledge that are needed by our FileDAG design.
	
	\noindent\textbf{Decentralized storage network (DSN).} DSNs aggregate storage offered by multiple independent storage providers and self-coordinate to provide reliable and secure global data storage and retrieval services to clients without relying on any trusted third party. Generally speaking, the workflow of a DSN consists of two phases: put and get. Users put their files into the storage network and also get files with valid access keys from the network. A DSN must guarantee data integrity, retrievability and fault tolerance. We explain two techniques heavily used in FileDAG, namely content identifier (CID) and Proof-of-Storage (PoS). CID, as a fingerprint, is a hash-based unique identifier that maps to a data chunk. In FileDAG, a client can generate CIDs for each original file or increment. PoS helps miners prove that they have stored files physically. In Filecoin, a miner has to periodically generate proofs to demonstrate that files are indeed locally stored on hardware, which mitigates Sybil attacks.
	
	\noindent\textbf{DAG-based blockchain.}
	A blockchain is a decentralized tamper-proof append-only ledger. Nodes in a blockchain network achieve consensus on the ledger using a consensus algorithm. According to the ledger structure, blockchains can be categorized as chain-based or DAG-based. For a chain-based ledger, transactions are packed into blocks. Each block is hash-chained to its previous block to ensure consistency and persistence \cite{garay2015bitcoin}. As there can only be one block at a block height, chain-based blockchains have weak concurrency. Bitcoin-NG \cite{bitcoinng} intends to improve concurrency by adding micro blocks alongside a main chain. However, this method does not fundamentally improve concurrency. Therefore, DAG-based blockchains emerge \cite{nxt, tangle, dagcoin}. In a DAG-based blockchain, each block (or transaction) can point to multiple previous blocks and form a directed acyclic graph (DAG).
	Filecoin makes use of tipset to increase network throughput, where a tipset is a set of blocks, and the blockchain in Filecoin is a chain of tipsets. Blocks in a tipset can point at multiple blocks in the previous tipset. As a result, blocks in Filecoin form a DAG. But tipset is not flexible enough to support file indexing; therefore we propose a two-layer DAG-based blockchain in FileDAG to address this issue.

	\section{FileDAG Design}
	\label{sec:filedag_design}
	In this section, we begin with the design objectives and overview of FileDAG and then describe its design details. 
	
	\subsection{Design Objectives and Strawman}
	\label{ss:overview}
	
	\noindent \textbf{Design Objectives.} We design FileDAG following three objectives: (1) Consistency. Honest nodes should agree on the same view of the blockchain ledger and the same set of proofs of storage. Deals of storage should be irreversible. (2) Deduplication. Files stored on a DSN can share common components, especially in a multi-version file system. FileDAG should use efficient deduplication methods to save storage space. (3) Fast put \& get. The design of FileDAG should consider both bandwidth usage and latency; any mechanism that can help to save storage space should not bring too much extra latency.
	The overall latency of putting and getting a file in FileDAG should be low despite spending time on the increment generation.
	
	\begin{figure}[htbp]
		\centerline{\includegraphics[width=0.46\textwidth]{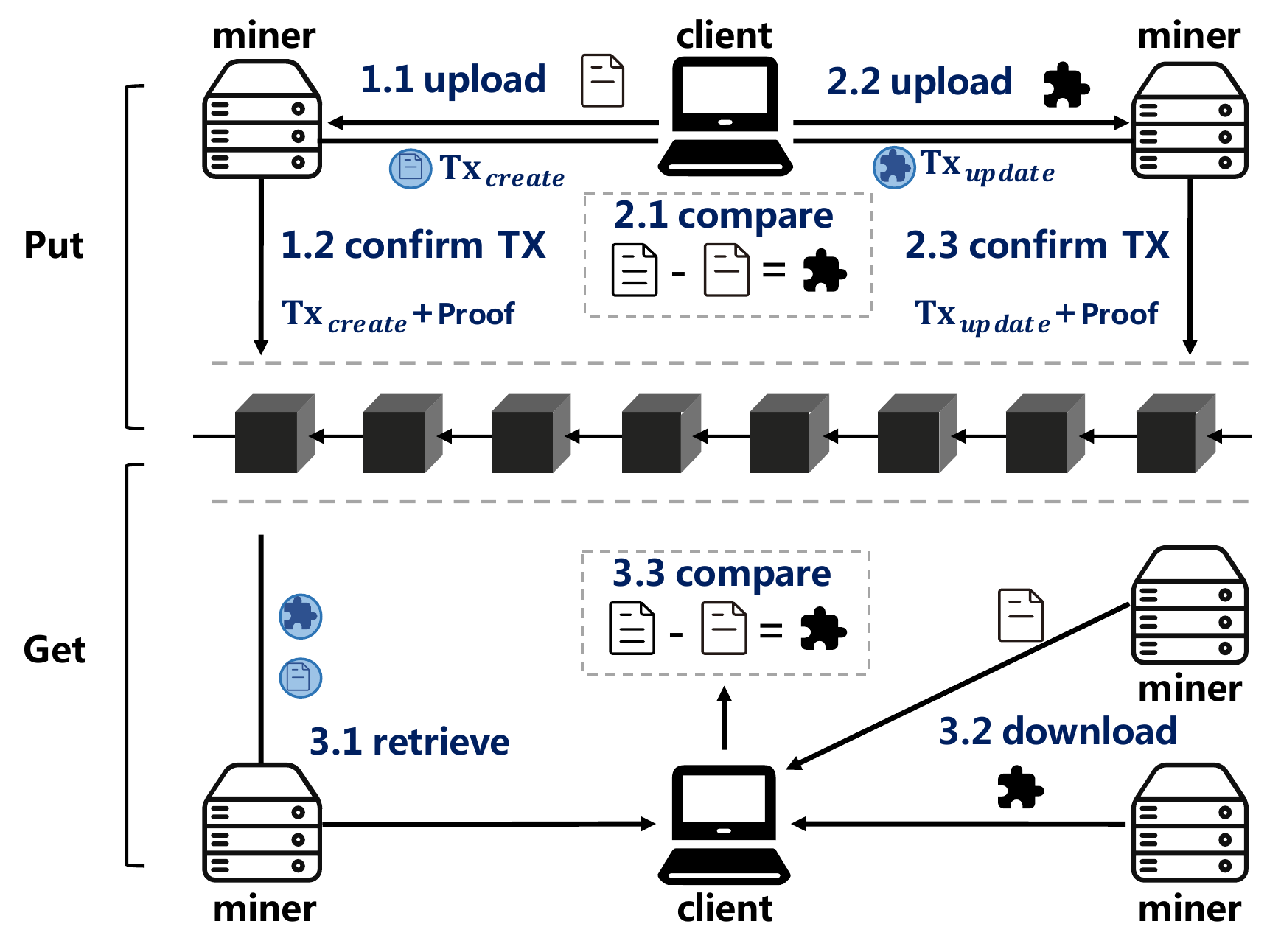}}
		\caption{A strawman design of FileDAG}
		\label{fig:strawman}
	\end{figure}
	
	\noindent \textbf{Strawman.} Here we provide a strawman as shown in Fig.~\ref{fig:strawman} to illustrate the whole picture of FileDAG. There are two entities in FileDAG, client and miner. Clients pay tokens to use storage, while miners earn tokens by providing services. Miners pledge storage to the FileDAG network to provide storage and retrieve services by helping clients search information on a blockchain. All miners maintain the blockchain ledger of FileDAG. 
	
	The workflow can be divided into two phases, put and get. In the put phase, a client can either create an original file or update an existing one on the FileDAG network. An original file should be uploaded to a miner [Step 1.1]. The miner then generates a  transaction $\mathsf{TX_{create}}$ for the file and broadcasts the transaction to the blockchain network [Step 1.2]. To update a file, the client first compares the new version to the previous version to get the increment [Step 2.1]; then the increment is sent to a miner who can issue the corresponding $\mathsf{TX_{update}}$ (or $\mathsf{TX_{merge}}$, $\mathsf{TX_{fork}}$). Note that only when transactions are confirmed on the blockchain can the put phase succeed. In the get phase, a client first sends a retrieval request [Step 3.1], then download the original file and the increments from the file holders to obtain a specific version [Step 3.2]; finally, the client assemble all fragments to recover the file [Step 3.3]. 
	
	In the following subsections, we detail the cores of FileDAG, including the increment generation, the two-layer DAG-based ledger, and the file recovery components. Note that our elaboration on increment generation focuses on the storage of multi-versioned files; but the idea is applicable to the more general case where a client specifies the relationship between two files, which is common in applications such as recreation of digital arts and quotes of contents. Table~\ref{tab:symboltabel} lists frequently used symbols to facilitate our presentation.
	
	\begin{table}[htbp]
	    \centering
	    \caption{Summary of Symbols}
	    \begin{tabular}{c l}
	    \toprule[1pt]
	    \textbf{Symbol} & \textbf{Description} \\
	    \midrule[0.5pt]
	        $G$ & DAG-based ledger of FileDAG \\
	        $V$ & the set of vertices in $G$ \\
	        $E_l$ & the set of edges in the lower layer of $G$\\
	        $E_u$ & the set of edges in the upper layer of $G$\\
	        $v$ & version of a multi-versioned file\\
	        $\Delta$ & increment \\
	        $\mathsf{TX}$ & transaction\\
	        $N$ & network size\\
	        $f$ & the maximum number of Byzantine fault nodes to tolerate\\
	        $\mathsf{CID}_v$ & Content ID of $v$\\
	        $\mathsf{REV}$ & revision operation\\
	        $\mathsf{ADD}$ & addition operation\\
	        $v_t$ & the $t$-th version of a multi-versioned file\\
	        $\tau_t$ & the type of operation that outputs $v_t$\\
	        $S_t$ & size of $v_t$\\
	        $I_t$ & size of the increment between $v_t$ and $v_{t-1}$\\
	        $E(\cdot)$ & expectation operator\\
	        $C$ & expected storage cost without increment-based storage\\
	        $C'$ & expected storage cost of increment-based storage\\
	        $n$ & the number of versions\\
	    \bottomrule[1pt]
	    \end{tabular}
	    \label{tab:symboltabel}
	\end{table}

	\subsection{Increment Generation}
	\label{ss:ig}
	Increment mechanism has been widely adopted in cloud computing to shorten backup windows and save storage \cite{nakivoincrement}. As we have discussed in our strawman design, FileDAG updates files by uploading increments instead of an entire new file. We propose an increment generation method based on our insight that files on DSNs are not simply static but changes over time. Such dynamicity can be found everywhere especially when storing codebases, medical records, mobile applications etc. Neighboring versions of a file usually share a large amount of duplicate contents. Our increment generation method intends to identify such contents, which later will be used for file recovery. 
	
	In concrete, FileDAG adopts patch algorithms to generate increments for multi-versioned files. To achieve a better performance, we adaptively use two patch algorithms, i.e., Myers \cite{myers1986ano} and BSDiff \cite{bsdiff}, to process text files and non-text files (binary files), respectively, rather than rely on one algorithm. Assume we have two files, an old one $A$ (of size $|A|$) and a new one $B$ (of size $|B|$). Both Git and diff commands in Linux use Myers, which takes $O(|A|+|B|+D^2)$ expected-time under a basic stochastic model \cite{myers1986ano}, where $D$ is the size of the minimum edit script between them. The Myers algorithm can quickly generate patches for text files, but cannot efficiently handle binary files. When forcing Myers to treat binary files as text files, the algorithm runs slowly and the complexity of generating an increment becomes $O(|A||B|)$.
    To process non-text files, we choose BSDiff which runs in $O((|A|+|B|)\log |A|)$ time. Besides, BSDiff has been widely used to generate patch files for mobile applications, which proves its effectiveness. In our implementation, FileDAG adaptively switches between Myers and BSDiff. It feeds the files into an increment module (see Fig.~\ref{fig:block_diagram}), which selects Myers for text files and BSDiff for non-text files to generate increments. In addition, we employ a small trick in which if $|\Delta_{AB}| > |B|$, FileDAG takes $B$ as a new original file instead of storing the increment $\Delta_{AB}$.
	
	To update a file, the client sends the increment to a miner who responds with a CID.  Then the miner generates a proof for this increment following the Proof-of-Storage protocol. In our implementation, we adopt the same PoS protocol as Filecoin since FileDAG does not focus on improving this process. 
	
	\subsection{Two-Layer DAG-based Ledger}
	\label{ss:dag}
	
	FileDAG uses a two-layer DAG-based ledger denoted as $G=(V, E_l, E_u)$. Both layers share the same set of vertices $V$. $E_l$ and $E_u$ are respectively the sets of edges in the lower layer and the upper layer. Each vertex in the ledger represents a transaction (a file version). Edges in the lower layer are used to describe derivative relationships between neighboring versions while the upper layer adds more edges to ensure consistency. 
	
	\noindent\textbf{Lower Layer $\bm{E_l}$.} Fig.~\ref{fig:version_dag} demonstrates an example lower layer $E_l$, in which each vertex (i.e., a transaction) also corresponds to a specific file version since it contains the CID of an original file or an increment. Each edge represents the derivative relationship between two vertices. 
	
	We allow four different types of transactions to describe operations launched by a client, including create, update, merge and fork, where the latter three depict the derivative relationships among different file versions. Correspondingly, an edge in $E_l$ represents an update, or a merge, or a fork operation.
    A create transaction is used to record a new original file. When a client creates a new file, it transfers the file to a miner, then constructs a create transaction $\mathsf{TX_{create}} \leftarrow \langle \mathsf{CREATE}, \mathsf{CID}_{v_0} \rangle$ and broadcasts it to blockchain, where $\mathsf{CID}_{v_0}$ is the identifier generated for the file. To verify that a transaction is sent by a client, each transaction should be correctly signed by the client's secret key. For convenience and clearance, we omit signatures when describing transactions in the rest of this paper. 
    When a client intends to put an increment to FileDAG, it sends an update transaction $\mathsf{TX_{update}} \leftarrow \langle \mathsf{UPDATE}, v, \mathsf{CID}_\Delta \rangle$, where $v$ is the version that the update follows, and $\mathsf{CID}_\Delta$ is the identifier of the increment. 

	A merge transaction $\mathsf{TX_{merge}} \leftarrow \langle \mathsf{MERGE}, v, v' \rangle$ combines two version branches $v$ and $v'$. A merge operation does not require adding new information therefore it does not generate increments. We only allow FileDAG to merge two versions at a time because merging multiple versions at once incurs a large complexity of addressing content conflict; but FileDAG can merge multiple versions by calling the merge operation multiple times. A fork transaction $\mathsf{TX_{fork}} \leftarrow \langle \mathsf{FORK}, v, V' \rangle$ can fork a version $v$ to get a set of new versions denoted by $V'$. Each new version contains an empty increment but is assigned a new CID. With fork operations, users can create and work on their own branches without the need of making new copies.

	\begin{figure}[!htbp]
		\centerline{\includegraphics[width=0.48\textwidth]{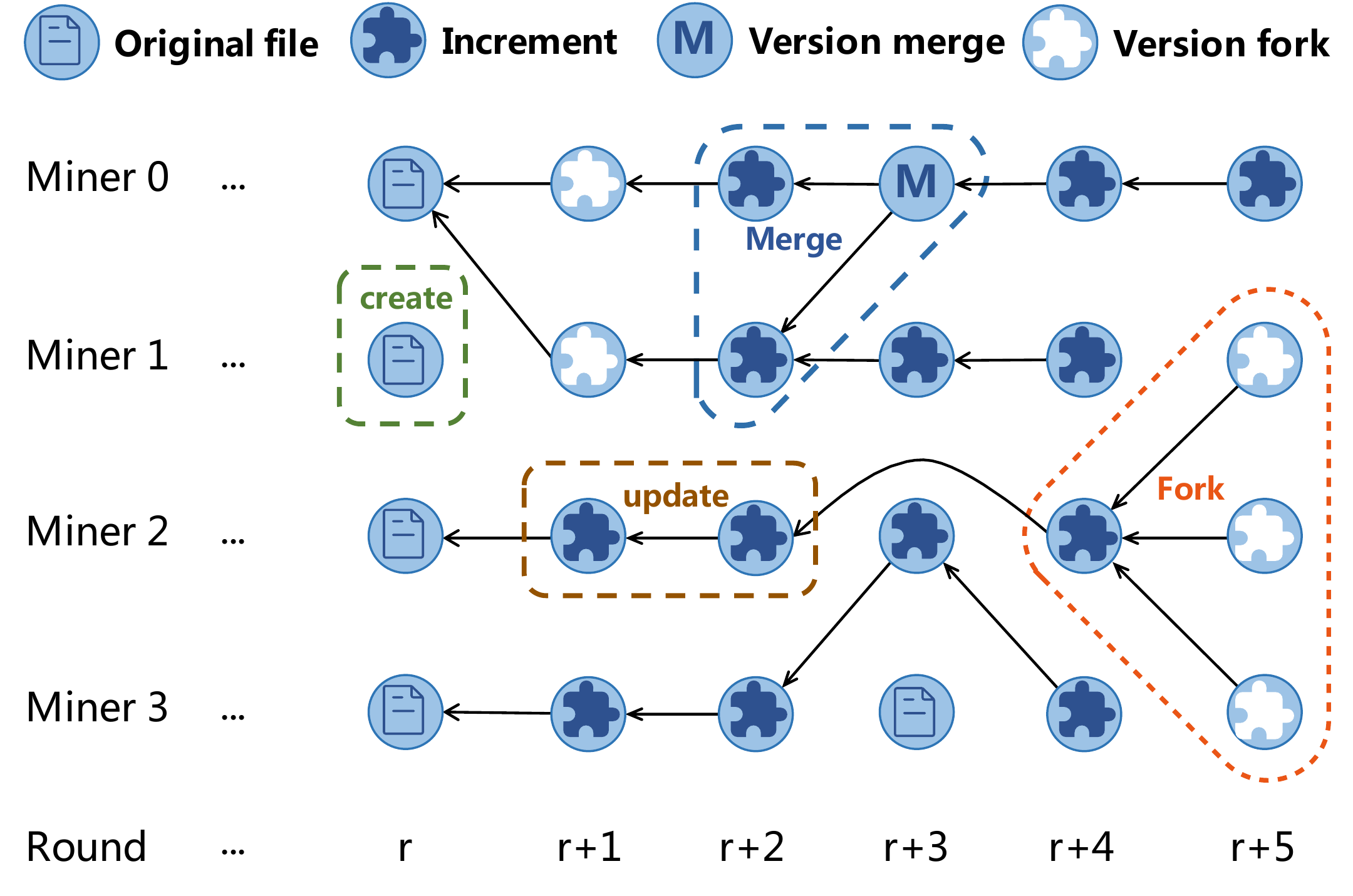}}
		\caption{Lower layer $E_l$ and the four types of transactions}
		\label{fig:version_dag}
	\end{figure}
	
\noindent\textbf{Upper Layer $\bm{E_u}$.} The DAG ledger formed by $E_l$ and $V$  has no consistency guarantee as it might be disconnected, making a miner unaware of a newly added transaction if the transaction is not linked to the ledger component stored by the minor. For example in Fig.~\ref{fig:version_dag}, Miner 0 does not know the newly added transaction created by Miner 3 at round $r+3$. This implies that honest miners may have different views of the ledger and fail to output the same result for a query, thus breaking the ledger's consistency property. 
To overcome such a problem, we add extra edges as shown in Fig.~\ref{fig:filedag_ledger} (the dotted arrows) to form the upper layer edge set $E_u$. More specifically, we modify the ledger construction algorithm in DAG-Rider \cite{dagrider} to construct $E_u$. In DAG-Rider, each vertex is associated with a round number (see Fig.~\ref{fig:filedag_ledger}). Each miner broadcasts one transaction (creating one vertex) per round and each vertex references at least $2f+1$ vertices in the previous round, where $f$ is the maximum number of Byzantine nodes to tolerate. That is, to advance to round $r+1$, a miner first needs to identify $2f+1$ vertices constructed by different miners at round $r$. Such a DAG construction is proved to achieve Byzantine atomic broadcast \cite{dagrider}, which possesses a strong consistency guarantee. Note that one can adopt other approaches to construct $E_u$, as long as the ledger formed by all  edges in $\mathrm{E}_l \cup \mathrm{E}_u$ realizes Byzantine atomic broadcast.
 More details about the Byzantine atomic broadcast will be discussed in Section~\ref{sec:analysis}.

The whole procedure of constructing our two-layer DAG-based ledger can be summarized as follows. At any round, an incoming transaction first points to those confirmed in the previous rounds, following the derivative relations (update, merge, or fork) to contribute edges to $E_l$; then we follow appropriate rules to select a number of other confirmed transactions and link the incoming transaction to them to construct edges for $E_u$. The DAG ledger formed by $\mathrm{E}_l \cup \mathrm{E}_u$  has strong consistency guarantee (see Section~\ref{sec:analysis}).
 
	\begin{figure}[htbp]
		\centerline{\includegraphics[width=0.48\textwidth]{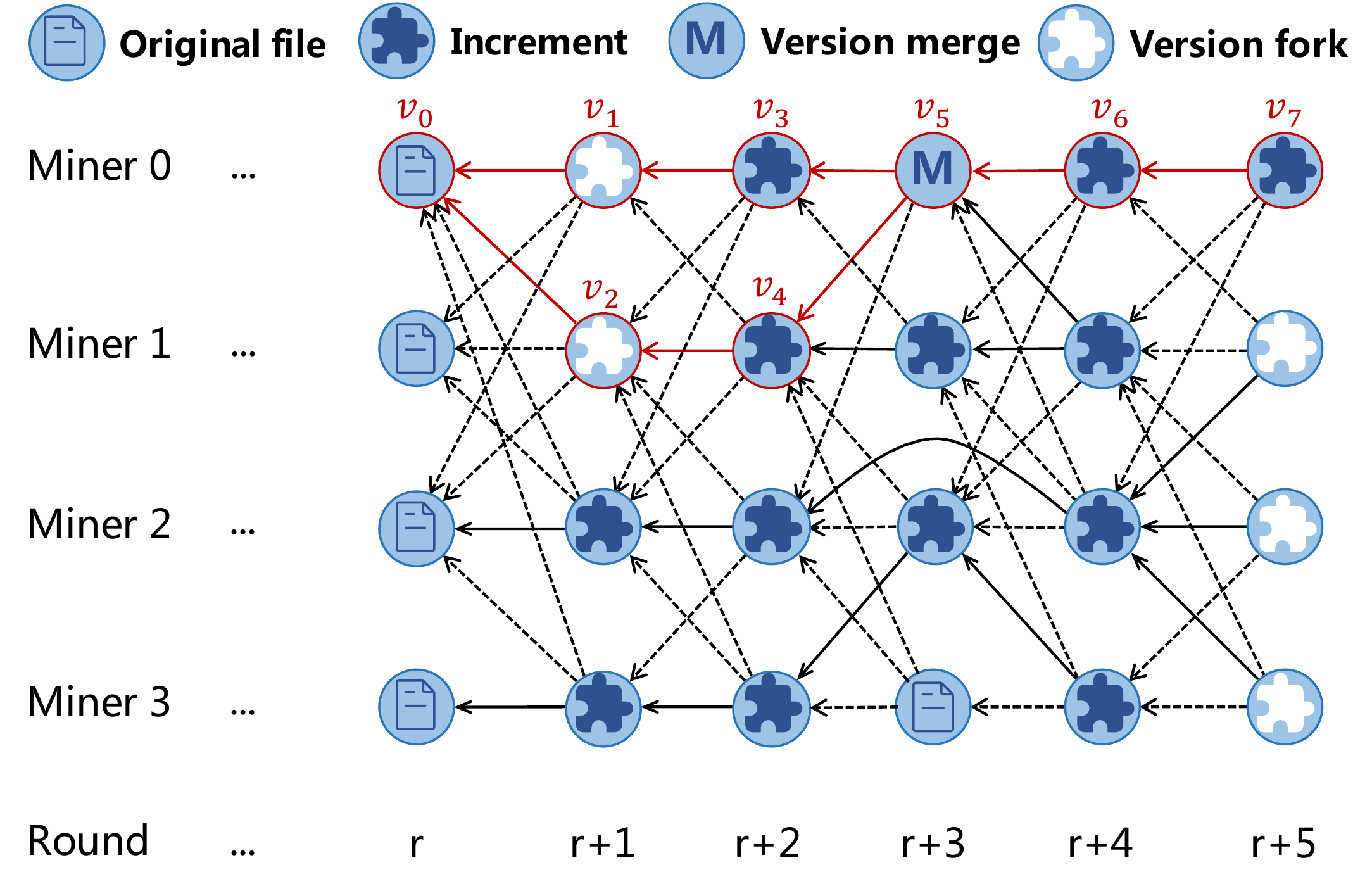}}
		\caption{Upper layer $E_u$}
		\label{fig:filedag_ledger}
	\end{figure}
	
	\subsection{File Recovery}
	\label{ss:filerecovery}
	Based on the lower layer of the DAG-based ledger, a miner can easily gather file fragments needed to recover a file. Considering that files are stored as increments for different versions, we propose Algorithm~\ref{alg:file_recover} to retrieve versions and recover the queried file. This algorithm consists of two functions, namely $\mathsf{Retrieve}()$ and $\mathsf{Recover}()$. 
	
	First, the miner runs $\mathsf{Retrieve}(v)$ (line 2-11) to obtain the versions of all fragments needed to recover the file of version $v$. 
	$\mathsf{Retrieve}(v)$ starts breadth first search (BFS) at version $v$ traversing the edges in $E_l$, and stops iteration when meets the version corresponding to a complete file. Using BFS rather than DFS (depth first search) can ensure that all the increments are marked in a reverse topological order without a sort procedure whose time complexity might be superlinear. All the traversed nodes are placed in the array $\mathsf{versions}$.  When BFS stops, $\mathsf{Retrieve}(v)$ reverses the order of the array $\mathsf{versions}$ and then returns it to the client (line 12). After receiving $\mathsf{versions}$, the client runs $\mathsf{Recover}(\mathsf{versions})$ to download the original file and the increments in $\mathsf{versions}$, patch the increments to the original file following the order in $\mathsf{versions}$, and finally output the requested file (line 15-20). The $\mathsf{Patch}$ algorithm takes either Myers or BSDiff, depending on whether the file is a text or not, as explained in section~\ref{ss:ig}.  These two algorithms both have time complexity of $O(|A|+D)$, where $|A|$ is the size of the file being patched and $D$ is the size of the increment. Thus, the overall time complexity of file recovery is linear to the total size of the original file and the increments.

	\begin{algorithm}
		\caption{File Recovery}
		\label{alg:file_recover}
		\textcolor{blue}{//Find vertices of a requested version on a ledger}\\
		\textbf{Function} $\mathsf{Retrieve}$($v$)\\
		$\mathsf{versions}[]$ $\leftarrow$ an empty array\\
		$Q$ $\leftarrow$ an empty queue\\
		$Q.\mathsf{Push}(v)$ \textcolor{gray}{//an empty queue to store versions}\\
		\While{$Q$ is not empty}{
			$\mathsf{temp}$ $\leftarrow$ $Q.\mathsf{Dequeue}()$\\
			$\mathsf{versions}.\mathsf{Append(temp)}$ \\
			\If{$\mathsf{temp}$.$\mathsf{type}$ $\ne$ $\mathsf{origin}$}{
				\For{each $\mathsf{pre}$ $\in$ $\mathsf{temp.previousVersions}$}{
					$Q.\mathsf{Push(pre)}$
				}
			}
		}
		Reverse and then return $\mathsf{versions}$\\
		
		\textcolor{blue}{//Recover a file using file fragments}\\
		\textbf{Function} $\mathsf{Recover(versions)}$\\
		Download all file fragments as $\mathsf{Data}$\\
		$v_0 = \mathsf{versions}[0]$\\
		$file$ $\leftarrow$ $\mathsf{Data}[v_0]$ \textcolor{gray}{//initialize file to an original one}\\
		\For{each $v$ $\in$ $\mathsf{versions}$ with $v$.$\mathsf{type}$=$\mathsf{increment}$}{
			$\mathsf{Patch}$($file$, $\mathsf{Data}[v]$)\\
		}
		return $file$
	\end{algorithm}

	\subsection{FileDAG Workflow}
	\label{ss:workflow}
	To end this section, we provide the workflow of FileDAG, which consists of five major steps including Create, Update, Retrieve, Download and Recover, as illustrated in Fig.~\ref{fig:protocol}. 
	
	\noindent\textbf{Create.} When creating an original file $v_0$ in the FileDAG network, a client first calculates $\mathsf{CID}_{v_0}$ as the fingerprint of ${v_0}$. The client then sends a message containing $\mathsf{CID}_{v_0}$ to a miner that might later provide storage services. If the miner is willing to store the file, it starts synchronizing $v_0$ with the client. After synchronization, the client signs and sends a create transaction $\langle \mathsf{CREATE}, \mathsf{CID}_{v_0} \rangle$ to the blockchain ledger. Then the miner generates a proof-of-storage for ${v_0}$ and settle down the received transaction. 
	
	\noindent\textbf{Update.} The main difference between creating and updating a file is that updating a file needs to store an increment rather than the entire complete file. Suppose we have a version $v$ and a new version $v'$. The increment generation method provides the client with an increment denoted by $\Delta$. Then the client calculates the CID of $\Delta$ and sends $\Delta$ to a storage miner who is responsible for generating a proof and settling down the transaction. Recall that when updating a file, a client can issue three types of transactions, namely update, merge and fork.
	
	\noindent\textbf{Retrieve.} Fetching a file with a specific version in FileDAG consists of two steps: Retrieve and Recover.  A client sends a retrieve message containing $\mathsf{CID}_{v}$ to a miner that provides retrieval services to get the list of CIDs to recover $v$. After receiving a retrieve message, the miner first looks up in its DAG ledger to locate $\mathsf{CID}_{v}$ and then calls function $\mathsf{Retrieve}(\mathsf{CID}_{v})$ and forwards its output $\mathsf{versions}$ to the client. In the example illustrated in Fig.~\ref{fig:protocol}, the CID list contains $\mathsf{CID}_v$ and $\mathsf{CID}_\Delta$.
	
	\noindent\textbf{Download \& Recover.} After obtaining $\mathsf{versions}$, the client calls function $\mathsf{Recover}(\mathsf{versions})$ to downloads all related file fragments based on $\mathsf{versions}$. For each file fragment, the client needs to send a download message along with the CID to the miner who stores the data. After gathering all the required components, the client patches the increments to the original file to recover file $v$.

	\begin{figure}[tb]
		\centerline{\includegraphics[width=0.48\textwidth]{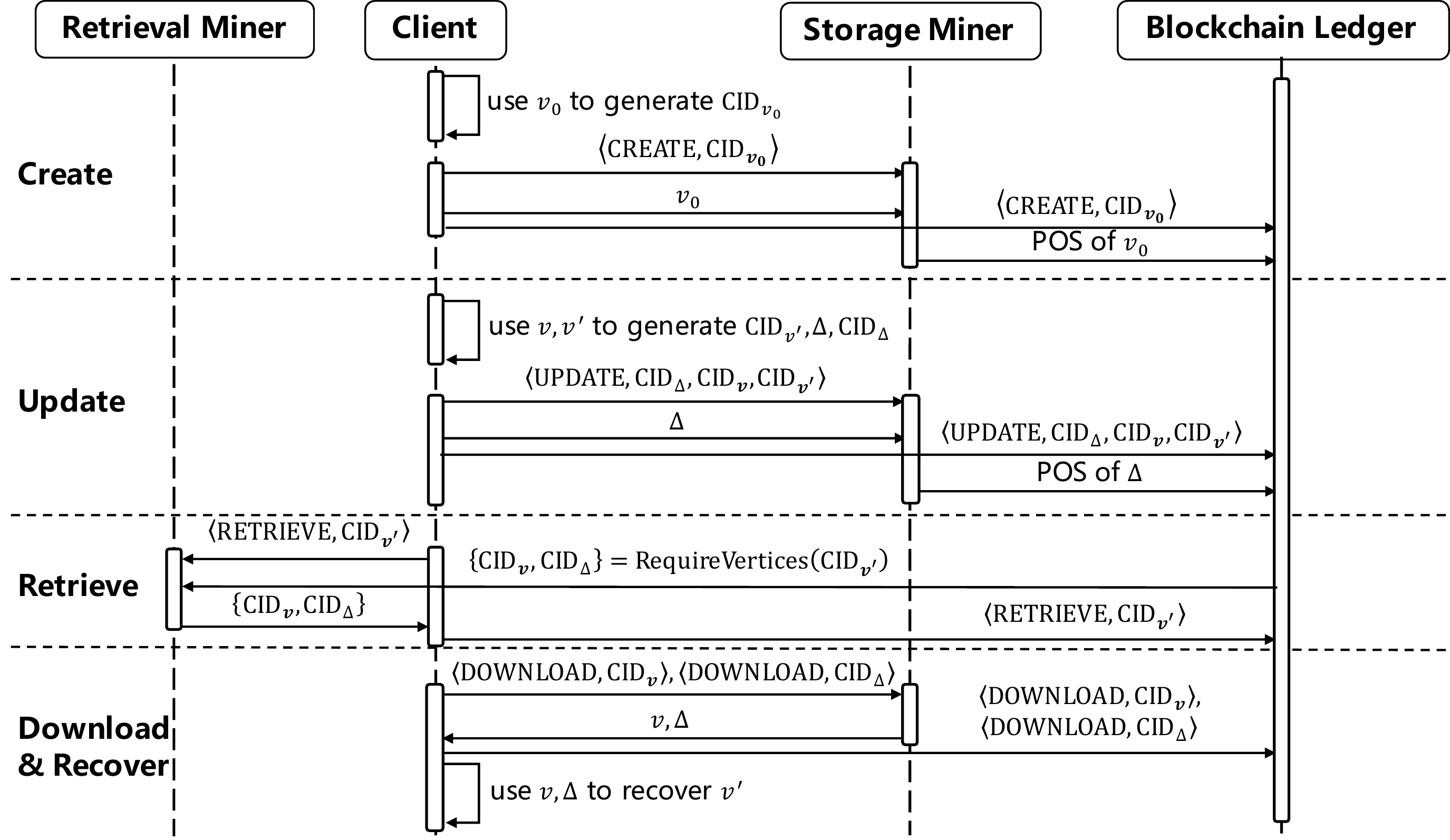}}
		\caption{Protocol sequence diagram of FileDAG.}
		\label{fig:protocol}
	\end{figure}
	
	\section{Analysis}
	\label{sec:analysis}
	
	In this section, we analyze two properties of FileDAG: consistency and efficiency (in terms of storage cost).
	
	\subsection{Consistency}

Byzantine atomic broadcast \cite{dagrider} guarantees the following properties:

\begin{itemize}
    \item \textbf{Agreement.} If an honest node $p$ commits vertex $i$ in a DAG, then every honest node $p'$ eventually commits $i$.
    \item \textbf{Integrity.} For each round $r$ and node $p$, an honest node $p'$ accepts at most one vertex proposed by $p$ in round $r$.
    \item \textbf{Validity.} If an honest node $p$ proposes a vertex $i$ in round $r$, then every honest node $p'$ eventually commits $i$.
    \item \textbf{Total order.} If an honest node $p$ commits vertex $i$ before committing vertex $j$, then no honest node commits vertex $j$ without first committing vertex $i$.
\end{itemize}

The construction of the two-layer DAG ledger ($\mathrm{E}_l \cup \mathrm{E}_u$) in FileDAG follows the algorithm of DAG-Rider; thus $\mathrm{E}_l \cup \mathrm{E}_u$ can be reduced to DAG-Rider's ledger. As DAG-Rider is a byzantine atomic broadcast implementation which guarantees the above properties, the two-layer DAG ledger of FileDAG possesses these properties as well. 	
	
	\begin{definition}
	    (Consistency of multi-version DSN). For any version $v$ of a file, an honest node can be convinced by the PoS proof that $v$ is available in the FileDAG network only when $v$ is indeed available; and if an honest node claims that $v$ is available, then all other honest nodes claim the same. 
	\end{definition}
	
	\begin{theorem}\label{theo:dsnconsistency}
	    FileDAG meets consistency of multi-version DSN.
	\end{theorem}
	
	\begin{proof}
	    First, the set of PoS committed in the FileDAG ledger is consistent. Based on the agreement and validity of byzantine atomic broadcast \cite{dagrider}, each PoS proposed by an honest node is committed by all honest nodes, and every honest nodes commit the same set of PoSes. 
	    
	    Second, the consistency of PoS verification, i.e., for any version $v$ of a file, if any honest node $p$ accepts that $v$ is available by verifying the PoS proofs, then every other honest node $p'$ accepts that $v$ is available if $p'$ verifies the corresponding proofs. FileDAG employs the PoS algorithm of Filecoin. Assuming the soundness of the PoS algorithm, i.e., a miner can output a valid PoS of a file if and only if it is able to output a copy of the file, every honest node outputs the same verification result for the same PoS. Provided the consistency of PoS committed in the FileDAG ledger, the array of CID $\mathsf{versions}$ output by function $\mathsf{Retrieve}(v)$ is determined for determined $v$, and thus the consistency of PoS verification is satisfied.
	    
	    Last, FileDAG meets consistency of multi-version DSN. Assume that the diff algorithms and the corresponding patch algorithms are correct so that when comparing two files (A and B), a diff algorithm always generates the same increment, and patching the increment to A always yields B. Combining the ledger consistency and the consistency of PoS verification, one can see that the consistency of multi-versioned files is proved.
	\end{proof}
	
	\subsection{Storage Cost}
	Next we analyze the storage cost of FileDAG. 
	Let $S_t$ denote the size of the $t$th version of a multi-versioned file and $I_t$ denote the size of the increment that the $t$th version differs from its previous version. After analyzing the growth of several GitHub repositories, we found that for each repository, typically there are two types of modifications for its growth, namely revision and addition. A revision ($\mathsf{REV}$) operation on a repository does not significantly change the size of the repository, and the size of increment between the updated version and its previous version is much smaller than that of the whole repository. An addition ($\mathsf{ADD}$) operation usually adds a large quantity of contents to a repository. For most repositories under our analysis, the number of revision operations is about ten or hundred times of that of the addition operations. But the size of the increment brought by an addition operation is ten or hundred times of the increment brought by a revision operation. To analyze the storage cost of FileDAG, we make a few assumptions on the growth of a multi-versioned file. 
	
	Consider the initial version of a file as an empty file with size zero, then the creation of a file can be regarded as an addition operation. In other words, we have $S_0 = 0$ and $S_1$ is the length of the original file.
	
    For each version $v_t, t>1$, the type $\tau_t$ of operation that outputs $v_t$ can be regarded as a random variable (and the sequence of operations observed by each node is consistent, as proved in Theorem~\ref{theo:dsnconsistency}). Let
	    \[\mathsf{Prob}(\tau_t=\mathsf{ADD}) = p, \mathsf{Prob}(\tau_t=\mathsf{REV}) = 1-p\]
	and
	\begin{equation}
		I_t=\left\{
		\begin{array}{rcl}
			r_t & & \tau_t=\mathsf{REV}\\
			a_t & & \tau_t=\mathsf{ADD}
		\end{array}
		\right.
	\end{equation}
	For a specific multi-versioned file, one can assume that the ratios $\frac{1-p}{p} < 1$ and $\frac{\mathrm{E}(r)}{\mathrm{E}(a)} > 1$ are constants. Then we have
	\[\frac{1-p}{p}\cdot\frac{\mathrm{E}(r)}{\mathrm{E}(a)} = O(1)\]
	
	\begin{theorem}
	    Let $C$ and $C'$ denote the expected storage cost of storing a file having $n$ versions without and with the increment mechanism, then we have $C' = O(n^{-1})C$.
	\end{theorem}
		
	\begin{proof}
	    \[\mathrm{E}(I_n) = p\mathrm{E}(a)+(1-p)\mathrm{E}(r),\]
	    \[\mathrm{E}(S_n) = \sum_{t=1}^{n}p\mathrm{E}(a) =np\mathrm{E}(a)\]
	    \[C=\sum_{t=1}^{n}\mathrm{E}(S_t)=\frac{n(n+1)}{2}p\mathrm{E}(a)=O(n)\mathrm{E}(S_n).\]
	    \[C'=\sum_{t=1}^{n}\mathrm{E}(I_t)=np\mathrm{E}(a)+n(1-p)\mathrm{E}(r).\]
	    Based on our assumption, 
	    \[\frac{1-p}{p}\cdot\frac{\mathrm{E}(r)}{\mathrm{E}(a)} = O(1)\]
	    then
	    \[C'=np\mathrm{E}(a)(1+O(1))=O(1)\mathrm{E}(S_n).\]
	    Therefor, the expected storage cost of FileDAG to keep all versions of a multi-versioned file is nearly equal to the size of the newest version of the file, while traditional solutions cost $O(n)$ times in expectation. 
	    \[C'=O(n^{-1})C.\]
	\end{proof}

	In section~\ref{sec:exp}, we report our experimental results to further support the above conclusion.
	
	\section{Performance Evaluation}
	\label{sec:exp}
	
	To evaluate the performance of FileDAG, we carry out real experiments. Specifically, we implement FileDAG based on the description in Section~\ref{sec:filedag_design}, and perform the full processes of put and get operations. 
	
		\begin{figure}[!htbp]
		\centerline{\includegraphics[width=0.48\textwidth]{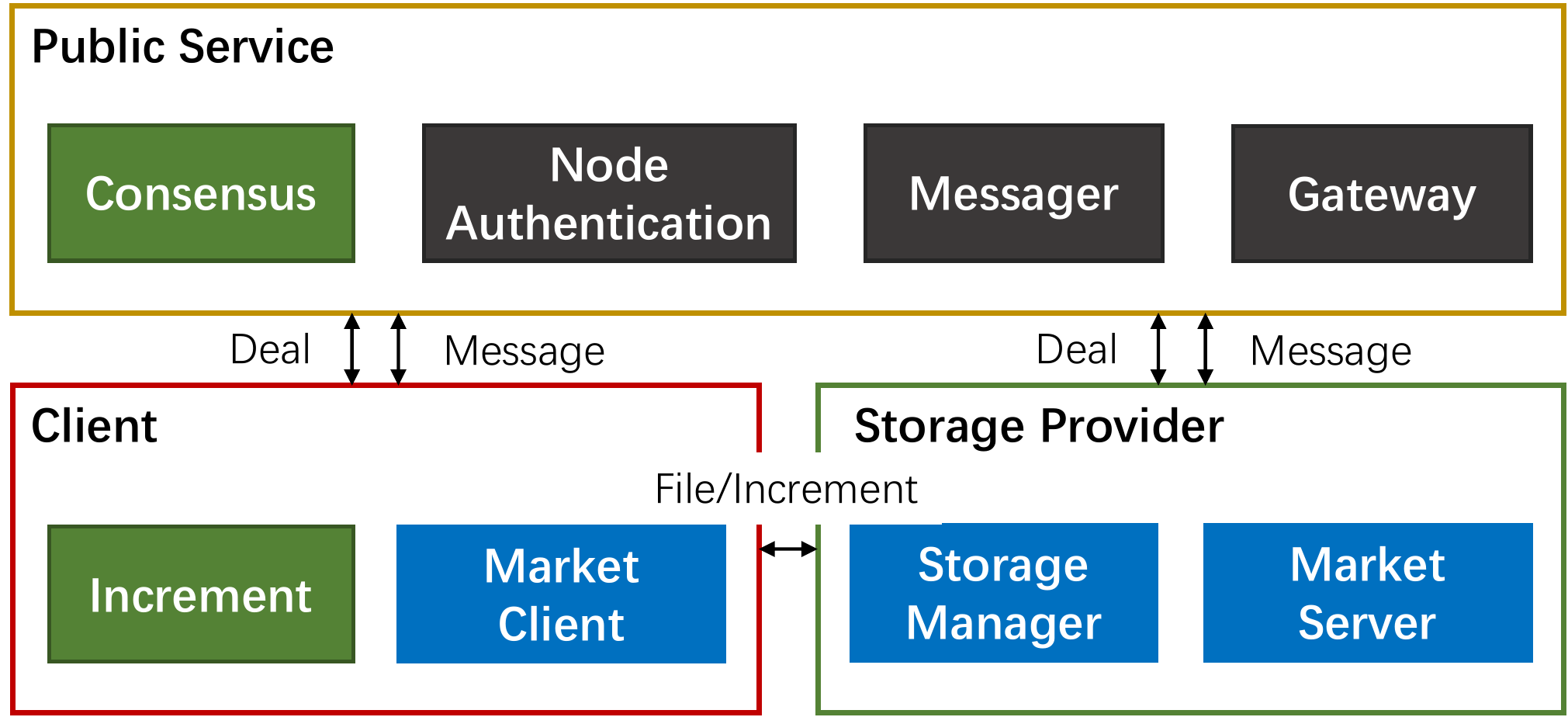}}
		\caption{Block diagram of FileDAG}
		\label{fig:block_diagram}
	\end{figure}
	
	\subsection{Implementation}
	
	As FileDAG shares many common properties with Filecoin, we implement it by making modifications on Filecoin. This also provides a chance for us to objectively compare FileDAG with Filecoin.
	Venus, which has a modular design and was written in Golang, is one of the four main implementations of Filecoin. We fork a branch from Venus (version v1.1.3-rc1) to build FileDAG. 
	A typical Venus deployment consists of three components, namely Public Service, Client, and Storage Provider (Storage Miner) \cite{venus}. Correspondingly, FileDAG has the same three components and the modules included in each component are illustrated in Fig.~\ref{fig:block_diagram}. Specifically, we build the Consensus and Increment modules from scratch, and they are marked green in Fig.~\ref{fig:block_diagram}. Additionally,  we make adaptations on the Market modules (Market Client and Market Server) and the Storage Manager of Venus and reuse them in FileDAG -- they are colored blue in Fig.~\ref{fig:block_diagram}. The three modules marked gray, i.e., Node Authentication, Messager, and Gateway,  are directly inherited from Venus. 
	The Increment module is implemented with 1096 lines of code in Golang and is included in Client. It consists of an Increment Generator and an Increment Accumulator, which can be deployed separately, as a data consumer who does not create or update files may not need the Increment Generator while a data provider who does not download files may not need the Increment Accumulator.
    We build a brand-new Consensus module to replace the one in Venus and develop our two-layer DAG ledger for FileDAG. The upper layer of the ledger is realized by an independent prototype of DAG-Rider with 424 lines of code in Golang. Note that the lower layer is implemented when revising the Market Client. 
	We also modify the Market Server in Storage Provider to attach version control information when transferring files/increments, and add the $\mathsf{Retrieve()}$ function to the Storage Manager module to obtain the versions of all fragments needed to recover the file of a particular version.
	Modifications on Venus take 603 lines of code in total. Other modules of Venus that are not mentioned in this paper stay as they are so that we can objectively evaluate the performance enhancement brought by our innovation. All components of FileDAG implementation are written in Golang, and we build and test FileDAG with go1.17.11.

	\begin{table*}[!htbp]
		\caption{Comparison between FileDAG and other state-of-the-art DSNs}
		\label{tab:eva}
		\begin{center}
			\begin{tabular}{l r r r r r r r r}
				\toprule[1pt]				
				\ & \multicolumn{4}{c}{Text} & \multicolumn{1}{c}{Multimedia} & \multicolumn{3}{c}{APK (Binary)} \\
				\cline{2-5} \cline{7-9}
				\ & \makecell{IPLD} & \makecell{go-ipfs} & \makecell{ccf-\\deadlines} & \makecell{Git} & \makecell{PPT} & \makecell{Minecraft} & \makecell{WeChat} & \makecell{Netflix} \\
				\midrule[0.8pt]
				\makecell[l]{\# of versions} & 212 & 129 & 244 & 845 & 21 & 277 & 18 & 20\\ 
				\makecell[l]{Average size (MB)} & 3.3 & 7.3 & 1.8 & 50.9 & 8.7 & 14.4 & 228.4 & 99.1\\ 
				\hline
				\makecell[l]{\textbf{Storage (FileDAG) (MB)}} & 7.1 & 16.3 & 2.9 & 104.8 & 11.8 & 1503.0 & 3052.3 & 356.8\\ 
				\makecell[l]{Storage (Filecoin (Venus)) (MB)} & 712.8 & 985.3 & 483.4 & 43035.7 & 183.7 & 3994.4 & 4110.4 & 1981.7\\ 
				\makecell[l]{Storage (Filecoin (Lotus)) (MB)} & 712.8 & 985.3 & 483.4 & 43035.7 & 183.7 & 3994.4 & 4110.4 & 1981.7\\ 
				\makecell[l]{Storage (Sia) (MB)} & 642005.1 & 47889.1 & 11850.5 & 2592801.5 & 391.6 & 12839.2 & 78046.4 & 47644.7\\ 
				\hline
				\makecell[l]{\textbf{Put runtime (FileDAG) (s)}}& 5.1 & 5.3 & 5.0 & 5.4 & 7.7 & 22.2 & 760.5 & 92.1\\ 
				\makecell[l]{Put runtime (Filecoin (Venus)) (s)} & 9.4 & 13.7 & 7.9 & 59.2 & 14.5 & 20.9 & 232.3 & 103.3\\ 
				\makecell[l]{Put runtime (Filecoin (Lotus)) (s)} & 9.3 & 13.8 & 8.2 & 59.0 & 14.7 & 20.3 & 234.2 & 106.3\\ 
				\makecell[l]{Put runtime (Sia) (s)} & 112.6 & 407.1 & 51.1 & 500.9 & 74.1 & 33.7 & 534.6 & 227.7\\ 
				\hline
				\makecell[l]{\textbf{Get runtime (FileDAG) (s)}} & 0.8 & 1.0 & 0.7 & 1.0 & 1.3 & 6.6 & 182.8 & 19.2\\ 
				\makecell[l]{Get runtime (Filecoin (Venus)) (s)} & 4.2 & 8.3 & 2.6 & 53.5 & 9.9 & 15.7 & 232.7 & 103.2\\ 
				\makecell[l]{Get runtime (Filecoin (Lotus)) (s)} & 4.0 & 8.4 & 2.7 & 53.9 & 9.6 & 15.5 & 230.9 & 104.1\\ 
				\makecell[l]{Get runtime (Sia) (s)} & 8.3 & 24.1 & 5.0 & 106.6 & 20.1 & 43.8 & 671.2 & 313.2\\ 
				\bottomrule[1pt]
			\end{tabular}
		\end{center}
	\end{table*}

	\subsection{Experiment Setup}
	\label{ss:exp_setup}
	
	We deploy FileDAG on 5 computers, with each having 2-Core CPU, 4GB memory and 40GB NVMe SSD, and running Ubuntu 22.04 LTS. The bandwidth of each computer is 1MB/s. We use the 5 computers to run: 1 Service node, 3 Storage Providers, and 1 Client.
	
	Based on the design difference between FileDAG and Filecoin, one can see that the performance changes brought by FileDAG come from three aspects: 1) storage space saved by the increment mechanism and the novel DAG ledger, 2) extra processing latency (in both uploading and downloading) brought by the increment mechanism, and 3) the decreased transmission latency due to smaller sizes of the payloads. We evaluate the storage cost and runtime of the operations in the put and get phases when providing multi-versioned file storage.	
	
	Files used in our evaluation consist of three types: text, multimedia, and binary. There exist plenty of online text files, e.g., code repositories, at GitHub.
    As shown in Table~\ref{tab:eva}, we clone 4 repositories from GitHub, namely IPLD, go-ipfs, ccf-deadlines, and Git. We extract all the versions in each repository, and store them in our implemented FileDAG network. Each of these repositories has hundreds of versions and their average sizes range from 1.8 MB to 50.9 MB. Additionally, we use FileDAG to store a multi-versioned presentation PPT (Microsoft PowerPoint) as an example of multimedia file. Such files are common in practice as when preparing academic reports or degree defenses, people usually make revisions on a PPT multiple times and keep historical versions for possible rolling backs. 
	In our evaluation, the PPT file used is a research report maintained by one of the authors. This file has 21 versions and the average size is 8.7 MB.  
	Finally, we take the APK (Android application package) files of several popular apps as examples of binary files. These apps include a game (Minecraft), an instant messaging software (WeChat), and an entertainment app (Netflix). We have downloaded these APK files from Uptodown\footnote{https://en.uptodown.com}, which provides downloads of APK files with different versions. As of this writing, WeChat and Netflix have 18 and 20 versions on Uptodown, while Minecraft has 277 versions. The average sizes of them range from 14.4 MB to 228.4 MB. We test the APK files because the BSDiff algorithm used in our Increment module has been widely adopted to generate patch files for APK updates in Android applications. 
	For comparison purpose, we select two market-tested DSNs,  i.e., Filecoin and Sia mentioned in Section~\ref{sec:related}, as the baselines for our evaluation. Particularly, we include two implementations of Filecoin, namely Venus and Lotus, for a more comprehensive comparison study.

	\subsection{Evaluation Results}
	
	A summary of the evaluation results averaged over 100 trials on the eight datasets mentioned above are reported in Table~\ref{tab:eva}, in which the three sections show the evaluation results in terms of storage, put (upload) runtime, and get (download) runtime. Storage measures the storage cost of the DSNs storing a multi-versioned file. The put/get runtime measures the average time cost of the DSNs to complete a put/get operation over one version of a multi-versioned file. One can see that FileDAG saves up to  $25\% \sim 99\%$ storage space compared to Filecoin. This is because FileDAG stores increments instead of the complete versions of the multi-versioned files.
    For the text and multimedia files, the put/get runtimes of FileDAG are significantly shorter than those incurred by other DSNs. However, limited by the performance of the increment generating algorithms processing binary files, the put runtime of FileDAG might be longer than those incurred by other DSNs, especially for WeChat, whose size is much bigger than those of the other two binary files. But this doesn't mean that FileDAG is not suitable for binary files. One can see that the get runtime of FileDAG is still significantly shorter than those of the other two Filecoins. As in practice, a binary file, e.g., a software installer, is usually downloaded multiple times (by different team members, or from different computers) once being put on the network, the longer time of an upload in FileDAG can be easily amortized by the shorter time of many downloads. 
    Besides, as shown in Table~\ref{tab:eva}, compared to FileDAG and Filecoin, Sia has extremely high storage cost and latency; thus it is ignored in the following studies. 
    
    \begin{figure*}[!ht]
	    \centering
		\begin{subfigure}{10cm}
			\includegraphics[width=\textwidth]{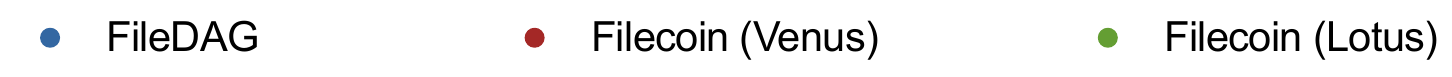}
		\end{subfigure}
		\qquad
		\\
		\begin{subfigure}{5.4cm}
			\includegraphics[width=\textwidth]{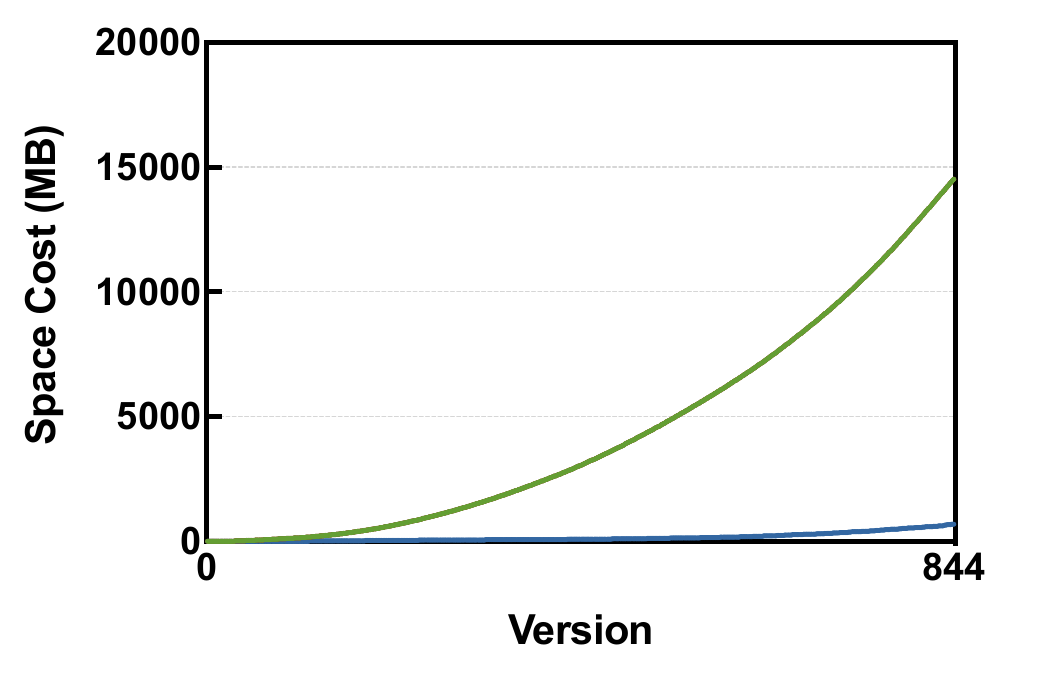}
			\caption{Git storage cost}\label{fig:textspace}
		\end{subfigure}
		\qquad
		\begin{subfigure}{5.4cm}
			\includegraphics[width=\textwidth]{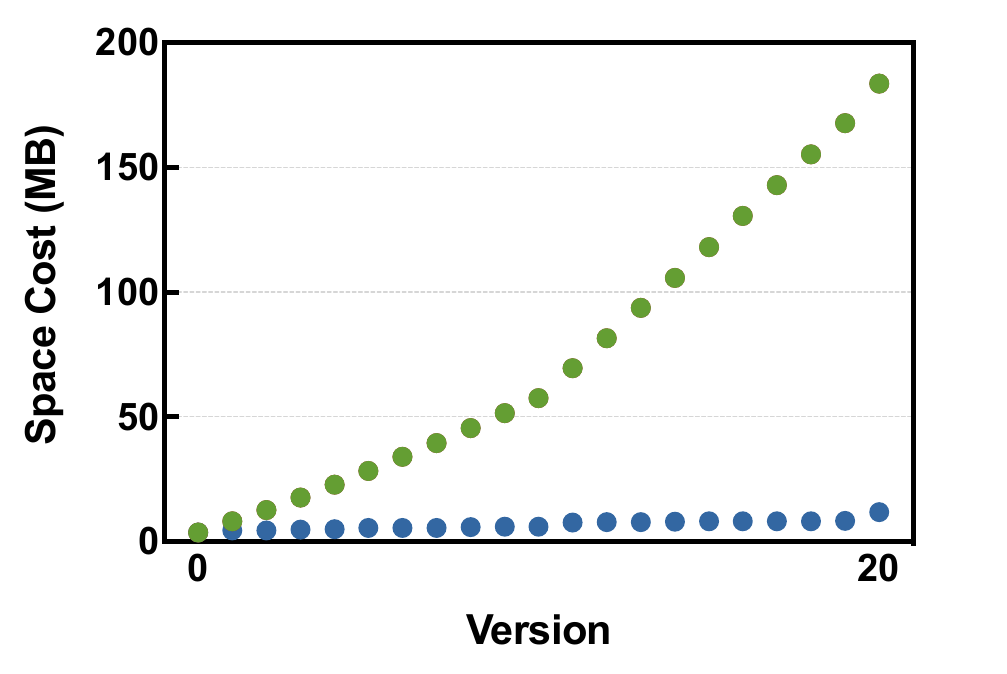}
			\caption{PPT storage cost}\label{fig:pptspace}
		\end{subfigure}
		\qquad
		\begin{subfigure}{5.4cm}
			\includegraphics[width=\textwidth]{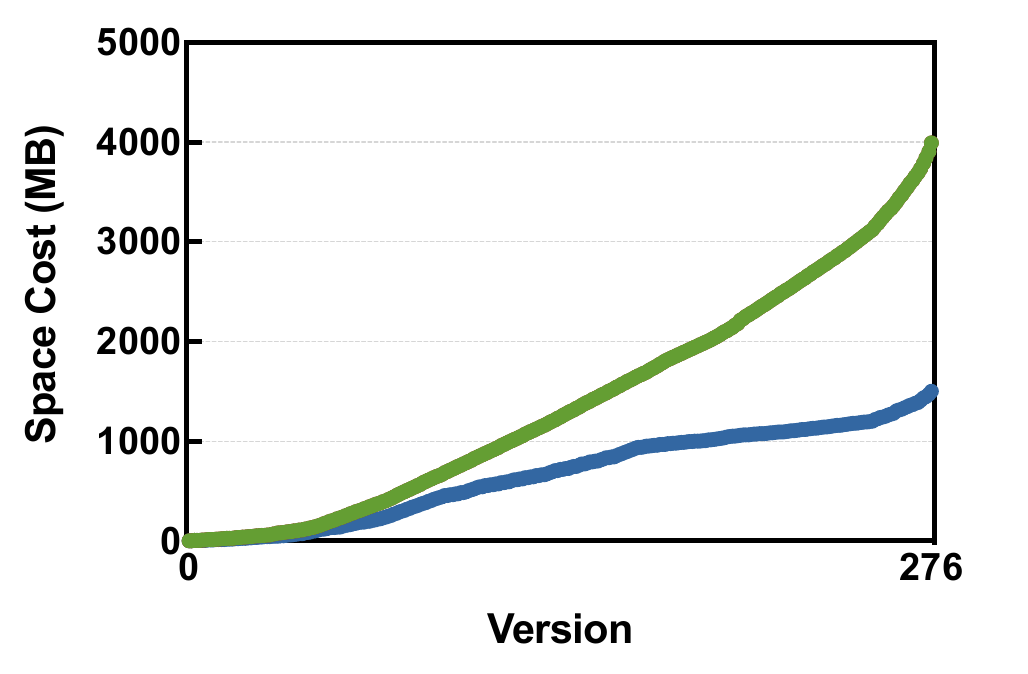}
			\caption{Minecraft storage cost}\label{fig:apkspace}
		\end{subfigure}
		\qquad
		\caption{Full Version Storage Costs}\label{fig:space_full}
	\end{figure*}
	
	\begin{figure*}[!htbp]
		\centering
		\begin{subfigure}{8cm}
			\includegraphics[width=\textwidth]{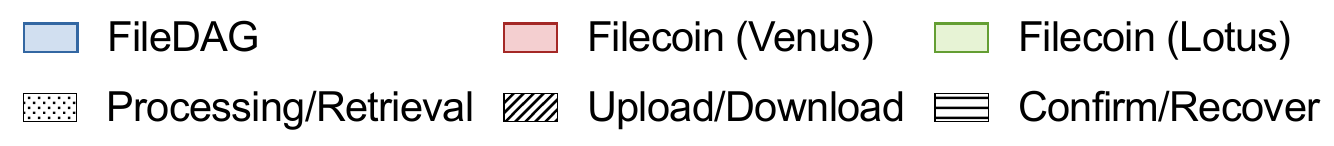}
		\end{subfigure}
		\\
		\begin{subfigure}{5.4cm}
			\includegraphics[width=\textwidth]{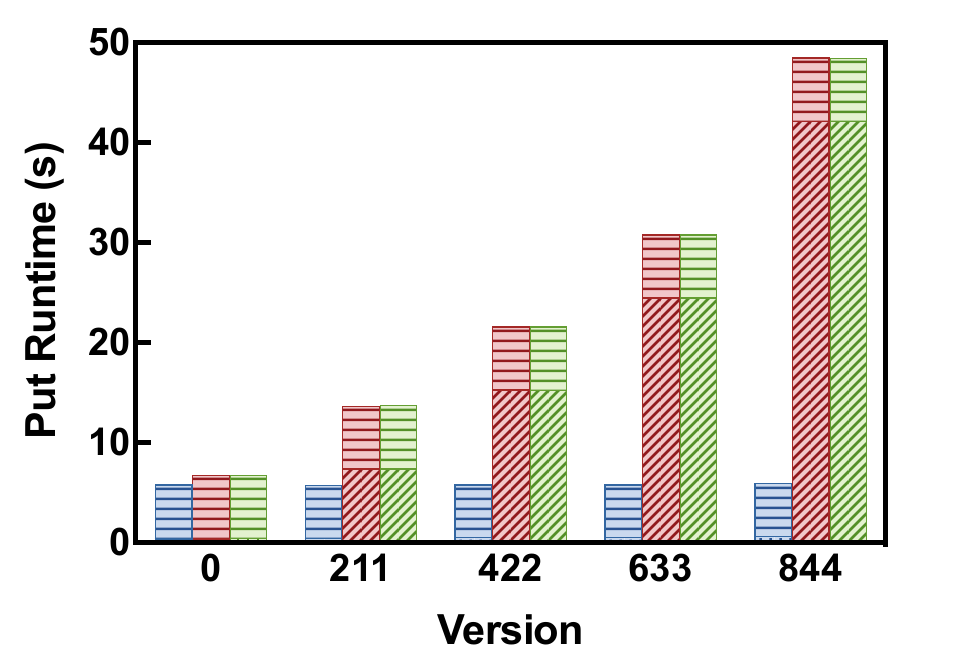}
			\caption{Put runtime of Git}\label{fig:gitup}
		\end{subfigure}
		\qquad
		\begin{subfigure}{5.4cm}
			\includegraphics[width=\textwidth]{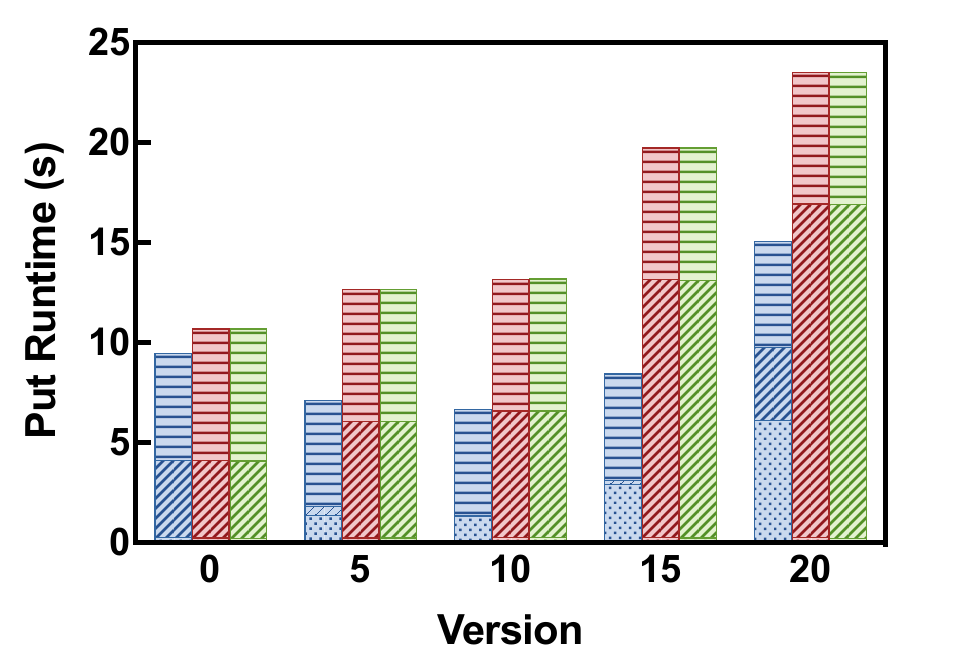}
			\caption{Put runtime of PPT}\label{fig:pptup}
		\end{subfigure}
		\qquad
		\begin{subfigure}{5.4cm}
			\includegraphics[width=\textwidth]{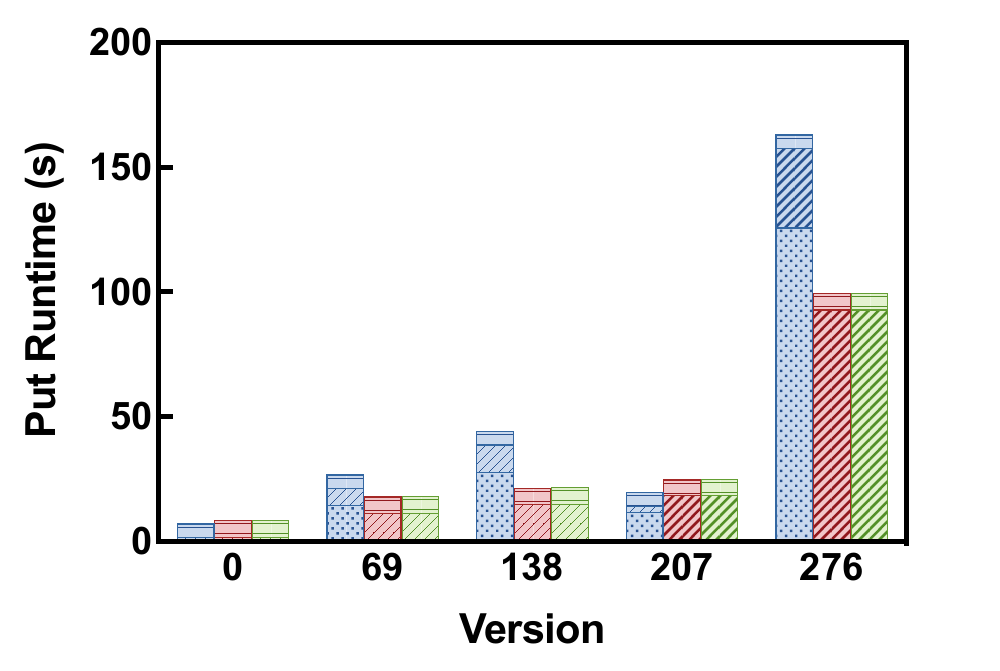}
			\caption{Put runtime of Minecraft}\label{fig:mcup}
		\end{subfigure}
		\qquad
		\\
		\begin{subfigure}{5.4cm}
			\includegraphics[width=\textwidth]{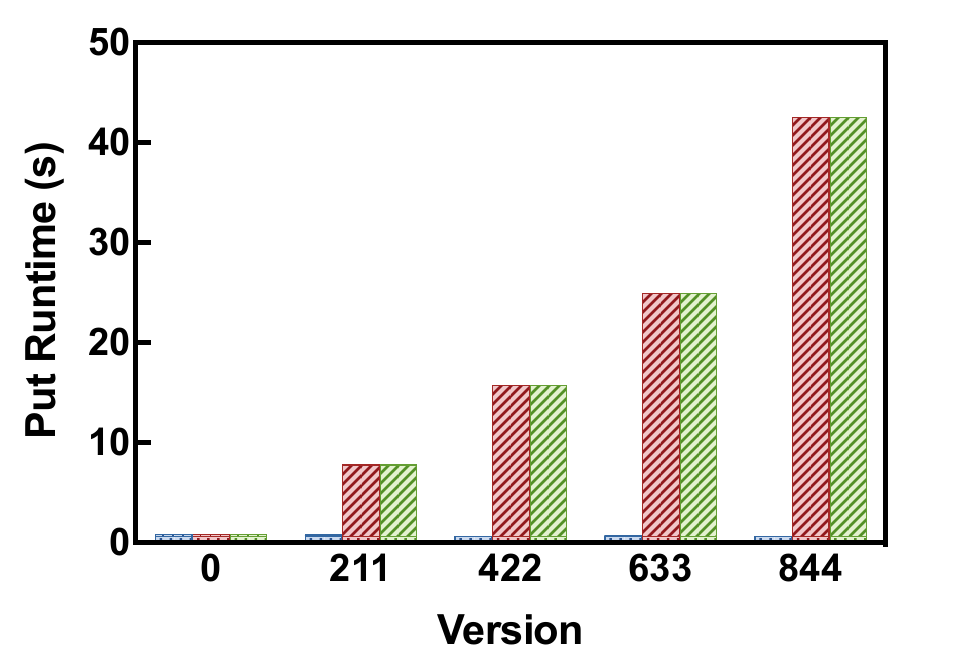}
			\caption{Get runtime of Git}\label{fig:gitdown}
		\end{subfigure}
		\qquad
		\begin{subfigure}{5.4cm}
			\includegraphics[width=\textwidth]{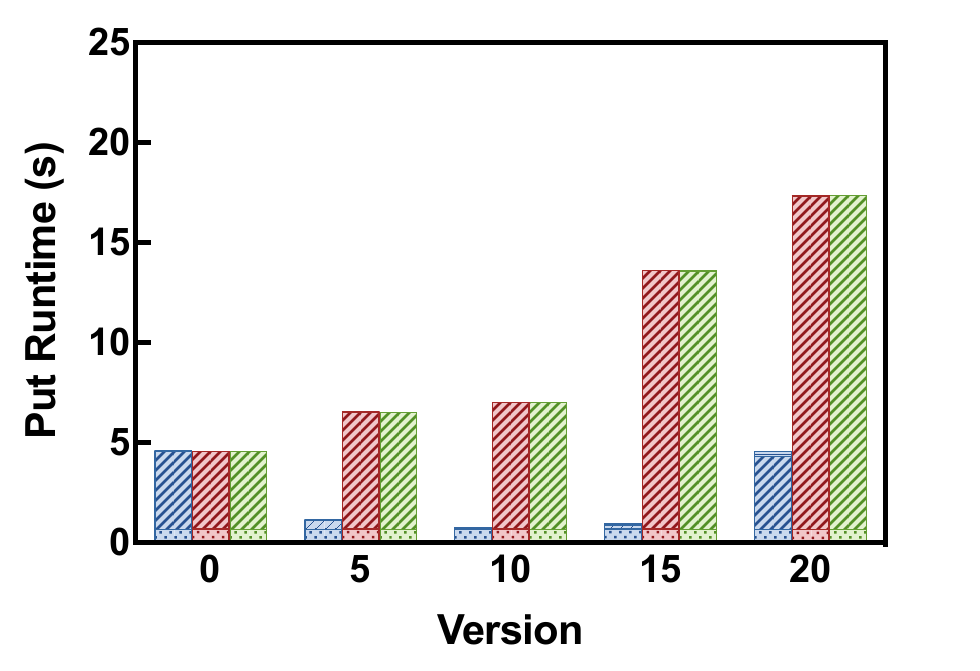}
			\caption{Get runtime of PPT}\label{fig:pptdown}
		\end{subfigure}
		\qquad
		\begin{subfigure}{5.4cm}
			\includegraphics[width=\textwidth]{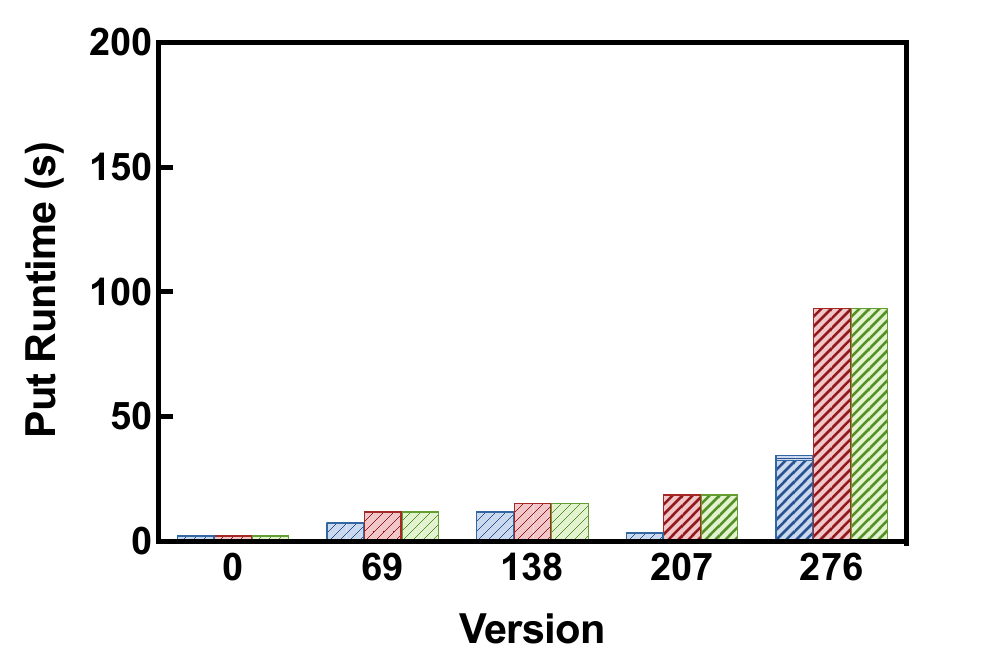}
			\caption{Get runtime of Minecraft}\label{fig:mcdown}
		\end{subfigure}
		\qquad
		\caption{Put and get runtime}
		\label{fig:latency}
	\end{figure*}

	To better demonstrate the superiority of FileDAG over the baseline DSNs, we use Fig.~\ref{fig:space_full} and Fig.~\ref{fig:latency} to respectively depict the accumulated storage cost and the put/get runtimes of the three types of files under our study. Due to limited space, we report the results of only one multi-versioned file in each category and select Git for text, PPT for multimedia, and Minecraft for binary. Particularly, in Fig.~\ref{fig:latency}, as it is impossible to draw the bar graphs of all versions, we choose 5 versions of each file, including the initial and final ones, to illustrate the results. 
 The experimental results of other files exhibit very similar trends. 
	
	\textbf{Storage Cost.} We measure and compare the disk usages of the DSN miners to demonstrate that FileDAG achieves low storage costs. As one can see from Fig.~\ref{fig:space_full}\footnote{Because Venus and Lotus adopt the same deterministic packing algorithm, these two DSNs take equal storage spaces; thus the curves for their storage cost results overlap in Fig.~\ref{fig:space_full}.}, while a multi-versioned file updates its versions, the storage size increases for all DSNs, but the growth of FileDAG is much slower compared to both implementations of Filecoin. More specifically, one can see that the storage costs of Filecoin exhibit roughly quadratic growth, while those of FileDAG show approximately linear growth. This confirms our analysis in Section~\ref{sec:analysis}. One can also see that the more version a file has, the more storage space FileDAG can save via the increment mechanism, and the size of the increment generated is related to how much a version differs from its previous version. This explains why the save of the storage space of Minecraft by FileDAG is not as big as those of the other two files, as versions of Minecraft are quite different.
	
	\textbf{Put Runtime.} The runtimes of the put phase of FileDAG and the baselines storing the above-mentioned three files are reported in the top three subfigures of Fig.~\ref{fig:latency}. In a typical DSN, the put phase consists of three steps: processing, transmission, and on-chain confirmation. During the processing step, a DSN client packs a new file version (in Filecoin) or the increment of the new version (in FileDAG) into specific format of payloads. Transmission gets the payloads transmitted from the client to the designated storage miner. After transmitting the file (or increment) to the miner, the client creates a transaction that contains the file CID and the miner ID, then gets it confirmed on the blockchain to complete the storage put operation. To better demonstrate these three steps, we split each bar in Fig.~\ref{fig:latency} into three sections with different textures. 
 
   As shown in Fig.~\ref{fig:gitup}, Fig.~\ref{fig:pptup} and Fig.~\ref{fig:mcup}, 
    one can see that the on-chain confirmation latency of all the tested files are similar and remain stable. This is because the transaction sizes remain stable no matter how large the files are as transactions only contain CIDs or hashes whose sizes are constants. Additionally,  based on our observation, the confirmation latency of FileDAG is about 1 second shorter than those of the two Filecoin implementations, thanks to the performance improvement brought by the consensus algorithm in our two-layer DAG ledger. Processing and transmission latencies intuitively depend on the sizes of the payloads. In Fig.~\ref{fig:pptup} and Fig.~\ref{fig:mcup}, one can see that due to increment generation, FileDAG has higher processing latencies and lower transmission latencies\footnote{The processing and transmission latences of Git in Fig.~\ref{fig:gitup} is not obvious as they are too short.}. In most cases, the put latencies of FileDAG are lower than those of the baseline. In some other cases such as for Minecraft, due to the complexity of the increment generation algorithm, the put latencies of FileDAG might be longer, but they are still acceptable, as usually put (upload) is an infrequent operation compared to the get (download) operation of the same file. 
    What's more, one can see that for the initial version of each file, FileDAG and Filecoin spend equal time on processing and transmission, because all DSNs process and transmit the same full-versioned file.
	
	\textbf{Get Runtime.} Similar to the put phase, the get phase of most DSNs also consists of three steps: retrieval, transmission, and file recovery. During retrieval, a miner gathers the required information (CIDs and the miner addresses) from the blockchain. Then the corresponding files (or increments) are sent to the client during the transmission step. As files can be encrypted for transmission or transmitted in small pieces, the client needs to recover the file in the last step. Therefore we also split each bar in Fig.~\ref{fig:gitup}, Fig.~\ref{fig:pptup}, and Fig.~\ref{fig:mcup} into three sections with different textures. One can see that the latency of a get operation is mainly attributed to retrieval and transmission, while file recovery latency is negligible. The retrieval latency of FileDAG and Filecoin are close and both stay stable at around 0.7 seconds. The transmission latency of FileDAG is lower because of the increment mechanism. Particularly, version 207 of Minecraft has a small update compared to version 206, so the corresponding increment is small, thus the transmission process finishes shortly. Note that the latency saved by FileDAG during the get phase is much more than that saved during the put phase, because the increment generation takes more time than recovery.  As we have mentioned earlier, practically a file is usually downloaded multiple times once being put on the network, the total latency saved by FileDAG can be significant.
	
	\section{Conclusions}
	\label{sec:conclusion}
	In this paper, we describe the design and implementation of FileDAG, a DSN built on DAG-based blockchain. FileDAG supports file-level deduplication in storing multi-versioned files by adopting an increment mechanism. Besides, FileDAG supports flexible and storage-saving file indexing by introducing a two-layer DAG-based blockchain ledger. We implement an actual instance of FileDAG and evaluate its performance. The results demonstrate that FileDAG outperforms the state-of-the-art industrial DSNs in storage cost and latency.
	In our future research, we plan to improve the performance of DSN by designing a faster and more effective proof of storage mechanism to further save storage cost of DSNs.
	
	
	\ifCLASSOPTIONcaptionsoff
	\newpage
	\fi
	
	\bibliographystyle{IEEEtran}
	\bibliography{FileDAG.bib}

\begin{thebibliography}{10}
\providecommand{\url}[1]{#1}
\csname url@samestyle\endcsname
\providecommand{\newblock}{\relax}
\providecommand{\bibinfo}[2]{#2}
\providecommand{\BIBentrySTDinterwordspacing}{\spaceskip=0pt\relax}
\providecommand{\BIBentryALTinterwordstretchfactor}{4}
\providecommand{\BIBentryALTinterwordspacing}{\spaceskip=\fontdimen2\font plus
\BIBentryALTinterwordstretchfactor\fontdimen3\font minus
  \fontdimen4\font\relax}
\providecommand{\BIBforeignlanguage}[2]{{%
\expandafter\ifx\csname l@#1\endcsname\relax
\typeout{** WARNING: IEEEtran.bst: No hyphenation pattern has been}%
\typeout{** loaded for the language `#1'. Using the pattern for}%
\typeout{** the default language instead.}%
\else
\language=\csname l@#1\endcsname
\fi
#2}}
\providecommand{\BIBdecl}{\relax}
\BIBdecl

\bibitem{spdl}
M.~Xu, Z.~Zou, Y.~Cheng, Q.~Hu, D.~Yu, and X.~Cheng, ``Spdl: A
  blockchain-enabled secure and privacy-preserving decentralized learning
  system,'' \emph{IEEE Transactions on Computers}, 2022.

\bibitem{tems}
C.~Liu, H.~Guo, M.~Xu, S.~Wang, D.~Yu, J.~Yu, and X.~Cheng, ``Extending
  on-chain trust to off-chain – trustworthy blockchain data collection using
  trusted execution environment (tee),'' \emph{IEEE Transactions on Computers},
  vol.~71, no.~12, pp. 3268--3280, 2022.

\bibitem{cloudchain}
M.~Xu, S.~Liu, D.~Yu, X.~Cheng, S.~Guo, and J.~Yu, ``Cloudchain: a cloud
  blockchain using shared memory consensus and rdma,'' \emph{IEEE Transactions
  on Computers}, 2022.

\bibitem{bittorrent}
B.~Cohen, ``The bittorrent protocol specification,'' 2008.

\bibitem{storageareanetwork}
J.~Tate, P.~Beck, H.~H. Ibarra, S.~Kumaravel, L.~Miklas \emph{et~al.},
  \emph{Introduction to storage area networks}.\hskip 1em plus 0.5em minus
  0.4em\relax IBM Redbooks, 2018.

\bibitem{filecoin}
P.~Labs. (2017) Filecoin: A decentralized storage network.

\bibitem{storj}
S.~Wilkinson, T.~Boshevski, J.~Brandoff, and V.~Buterin, ``Storj a peer-to-peer
  cloud storage network,'' 2014.

\bibitem{sia}
D.~Vorick and L.~Champine, ``Sia: Simple decentralized storage,''
  \emph{Retrieved May}, vol.~8, p. 2018, 2014.

\bibitem{swarm}
\BIBentryALTinterwordspacing
the Swarm~team. (2021) Swarm: Storage and communication infrastructure for a
  self-sovereign digital society. [Online]. Available:
  \url{https://www.ethswarm.org/swarm-whitepaper.pdf}
\BIBentrySTDinterwordspacing

\bibitem{w3s}
\BIBentryALTinterwordspacing
P.~Labs. (2022) Web3 storage - the simple file storage service for ipfs \&
  filecoin. [Online]. Available: \url{https://web3.storage/docs/}
\BIBentrySTDinterwordspacing

\bibitem{btfs}
\BIBentryALTinterwordspacing
BitTorrent. (2022) Btfs. [Online]. Available: \url{https://docs.btfs.io/}
\BIBentrySTDinterwordspacing

\bibitem{filebase}
\BIBentryALTinterwordspacing
Filebase. (2022) Use filebase as the origin for your cdn. [Online]. Available:
  \url{https://filebase.com/solutions/content-delivery/}
\BIBentrySTDinterwordspacing

\bibitem{merkledag}
\BIBentryALTinterwordspacing
P.~Labs. (2022) Merkle dags. [Online]. Available:
  \url{https://docs.ipfs.tech/concepts/merkle-dag/}
\BIBentrySTDinterwordspacing

\bibitem{versiongraph}
R.~Achar, \emph{The Global Object Tracker: Decentralized Version Control for
  Replicated Objects}.\hskip 1em plus 0.5em minus 0.4em\relax University of
  California, Irvine, 2020.

\bibitem{githubgraph}
\BIBentryALTinterwordspacing
GitHub. (2022) Configuring branches and merges in your repository. [Online].
  Available:
  \url{https://docs.github.com/en/repositories/configuring-branches-and-merges-in-your-repository}
\BIBentrySTDinterwordspacing

\bibitem{tangle}
S.~Popov, ``The tangle,'' \emph{White paper}, vol.~1, no.~3, 2018.

\bibitem{hashgraph}
L.~Baird, ``The swirlds hashgraph consensus algorithm: Fair, fast, byzantine
  fault tolerance,'' \emph{Swirlds Tech Reports SWIRLDS-TR-2016-01, Tech. Rep},
  vol.~34, 2016.

\bibitem{ipfs}
J.~Benet, ``Ipfs-content addressed, versioned, p2p file system (draft 3),''
  \emph{arXiv preprint arXiv:1407.3561}, 2014.

\bibitem{ethereum}
G.~Wood \emph{et~al.}, ``Ethereum: A secure decentralised generalised
  transaction ledger,'' \emph{Ethereum project yellow paper}, vol. 151, no.
  2014, pp. 1--32, 2014.

\bibitem{threefish}
N.~Ferguson, S.~Lucks, B.~Schneier, D.~Whiting, M.~Bellare, T.~Kohno,
  J.~Callas, and J.~Walker, ``The skein hash function family,''
  \emph{Submission to NIST (round 3)}, vol.~7, no. 7.5, p.~3, 2010.

\bibitem{gfs}
S.~Ghemawat, H.~Gobioff, and S.-T. Leung, ``The google file system,'' in
  \emph{Proceedings of the nineteenth ACM symposium on Operating systems
  principles}, 2003, pp. 29--43.

\bibitem{gnutella}
\BIBentryALTinterwordspacing
(2022) Regarding gnutella. [Online]. Available:
  \url{https://www.gnu.org/philosophy/gnutella.html}
\BIBentrySTDinterwordspacing

\bibitem{coralcdn}
\BIBentryALTinterwordspacing
(2022) Coral cdn. [Online]. Available: \url{http://www.coralcdn.org/}
\BIBentrySTDinterwordspacing

\bibitem{pass}
K.-K. Muniswamy-Reddy, D.~A. Holland, U.~Braun, and M.~I. Seltzer,
  ``Provenance-aware storage systems.'' in \emph{Usenix annual technical
  conference, general track}, 2006, pp. 43--56.

\bibitem{provchain}
X.~Liang, S.~Shetty, D.~Tosh, C.~Kamhoua, K.~Kwiat, and L.~Njilla, ``Provchain:
  A blockchain-based data provenance architecture in cloud environment with
  enhanced privacy and availability,'' in \emph{2017 17th IEEE/ACM
  International Symposium on Cluster, Cloud and Grid Computing (CCGRID)}.\hskip
  1em plus 0.5em minus 0.4em\relax IEEE, 2017, pp. 468--477.

\bibitem{garay2015bitcoin}
J.~Garay, A.~Kiayias, and N.~Leonardos, ``The bitcoin backbone protocol:
  Analysis and applications,'' in \emph{Annual international conference on the
  theory and applications of cryptographic techniques}.\hskip 1em plus 0.5em
  minus 0.4em\relax Springer, 2015, pp. 281--310.

\bibitem{bitcoinng}
I.~Eyal, A.~E. Gencer, E.~G. Sirer, and R.~Van~Renesse, ``$\{$Bitcoin-NG$\}$: A
  scalable blockchain protocol,'' in \emph{13th USENIX symposium on networked
  systems design and implementation (NSDI 16)}, 2016, pp. 45--59.

\bibitem{nxt}
\BIBentryALTinterwordspacing
N.~community. (2022) Nxt whitepaper. [Online]. Available:
  \url{https://nxtdocs.jelurida.com/Nxt_Whitepaper}
\BIBentrySTDinterwordspacing

\bibitem{dagcoin}
Y.~Ribero and D.~Raissar, ``Dagcoin whitepaper,'' \emph{Whitepaper, no. May},
  pp. 1--71, 2018.

\bibitem{nakivoincrement}
\BIBentryALTinterwordspacing
NAKIVO. (2021) Incremental backup. [Online]. Available:
  \url{https://www.nakivo.com/incremental-backup/}
\BIBentrySTDinterwordspacing

\bibitem{myers1986ano}
E.~W. Myers, ``An o(nd) difference algorithm and its variations,''
  \emph{Algorithmica}, vol.~1, no. 1-4, pp. 251--266, 1986.

\bibitem{bsdiff}
C.~Percival, ``Na{\i}ve differences of executable code,'' 2003.

\bibitem{dagrider}
\BIBentryALTinterwordspacing
I.~Keidar, E.~Kokoris-Kogias, O.~Naor, and A.~Spiegelman, ``All you need is
  dag,'' in \emph{Proceedings of the 2021 ACM Symposium on Principles of
  Distributed Computing}, ser. PODC'21.\hskip 1em plus 0.5em minus 0.4em\relax
  New York, NY, USA: Association for Computing Machinery, 2021, pp. 165--175.
  [Online]. Available: \url{https://doi.org/10.1145/3465084.3467905}
\BIBentrySTDinterwordspacing

\bibitem{venus}
\BIBentryALTinterwordspacing
I.~Force. (2022) Venus docs. [Online]. Available:
  \url{https://venus.filecoin.io}
\BIBentrySTDinterwordspacing

\end{thebibliography}

\end{document}